%% file: main.tex
\documentclass[11pt]{article}
\usepackage[utf8]{inputenc}
\usepackage{amsfonts}
\usepackage{amsmath}
\usepackage{amsthm}
\usepackage{amssymb}
\usepackage{natbib}
\usepackage{bbm}
\usepackage{standalone}
\usepackage[margin=1in]{geometry}
\usepackage{subcaption}
\usepackage{booktabs} 
\usepackage{multirow}
\usepackage{tikz}
\usepackage[linesnumbered, ruled, vlined]{algorithm2e}
\usepackage{placeins}
\usepackage{setspace}
\usepackage{mathtools}

\makeatletter
\newcommand{\removelatexerror}{\let\@latex@error\@gobble}
\makeatother

\DeclareMathOperator{\Prob}{P}

\newtheorem{theorem}{Theorem}[section]
\newtheorem{corollary}{Corollary}[theorem]
\newtheorem{lemma}[theorem]{Lemma}
\newtheorem{proposition}[theorem]{Proposition}
\newtheorem{definition}[theorem]{Definition}

\newcommand*\diff{\mathop{}\!\mathrm{d}}
\newcommand{\degr}{d}
\newcommand{\str}{s}

\title{Approximate Conditional Sampling for Pattern Detection in Weighted Networks}
\author{James A. Scott \\ Imperial College London \and Axel Gandy \\ Imperial College London}
\date{August 2021}

\begin{document}

\maketitle

\bigskip 
\begin{abstract}
Assessing the statistical significance of network patterns is crucial for understanding whether such patterns indicate the presence of interesting network phenomena, or whether they simply result from less interesting processes, such as nodal-heterogeneity. Typically, significance is computed with reference to a \textit{null model}. While there has been extensive research into such null models for unweighted graphs, little has been done for the weighted case. This article suggests a null model for weighted graphs. The model fixes node strengths exactly, and approximately fixes node degrees. A novel MCMC algorithm is proposed for sampling the model, and its stochastic stability is considered. We show empirically that the model compares favorably to alternatives, particularly when network patterns are subtle. We show how the algorithm can be used to evaluate the statistical significance of community structure.
\end{abstract}

\noindent%
{\it Keywords:} Approximate Tests, Exact Test, Null Models, Pattern Detection, Weighted Networks

\onehalfspacing

\section{Introduction}

This article develops a principled approach to assessing the significance of patterns observed in weighted networks. The proposed method compares the network of interest to graphs drawn from a null model. The null model is designed to account for node heterogeneity including both heavy-tailed degree and strength distributions. Unknown nuisance parameters are dealt with by approximate conditioning, and samples are drawn using a novel MCMC method. The development just outlined mirrors approaches that have long been used successfully to detect patterns in unweighted graphs \citep{connor_assembly_1979,newman_random_2001,milo_2002,maslov_detection_2004,stouffer_evidence_2007}.

Take, for example, the task of detecting community structure in weighted networks. In general, the community membership of nodes is unknown and must be recovered. Typically this is done by optimising some criterion, which could be a quality function like modularity \citep{newman_2004}. An alternative approach is to fit a statistical model which permits community structure using maximum likelihood. Possible models include the the stochastic block model \citep{nowicki_estimation_2001} and the degree-corrected stochastic block model \citep{karrer_stochastic_2011}.

Although popular, modularity is not based on any notion of the significance of a partition; rather it is defined as the absolute difference between observed inter-community links, and those expected under a given null model. As a result, it suffers from the \textit{resolution limit} \citep{Fortunato2007, kumpula_limited_2007}, whereby smaller modules cannot be detected in large networks. A number of methods attempt to overcome this by explicitly defining notions of significance \citep{Aldecoa2011,Miyauchi2016, Traag2013, Reichardt2006, palowitch_significance-based_2018,he_computing_2020}, which can be optimised over network partitions. 

Nonetheless, these methods consider the $p$-value of a fixed partition and are invalid when assessing the significance of a partition which results from optimising an objective function. It is possible to find partitions with low p-values in random graphs with no embedded community structure \citep{guimera_modularity_2004, Reichardt2006,fortunato_community_2010}.  The p-values are incorrect unless they account for the selection process. This phenomenon parallels that of inference post model selection, which is a widely studied problem that has recently garnered much attention within the statistics community \citep{taylor_statistical_2015,hastie_statistical_2019}.

In this article, we introduce a null model which can be used to quantify the significance of general patterns found in weighted graphs. For example, the approach can be used to determine the significance of community structure \textit{after having identified a partition} with an optimisation method. The approach is based on a generalisation of `rewiring' Markov chains \citep{ryser_1963, hakimi_realizability_1962, rao_1996} to weighted graphs, and is inspired by a recently developed Markov chain \citep{gandy_2016} for weighted graphs.

After introducing terminology, Section \ref{sec:wg_motiv} motivates the problem by first reviewing a common approaches in the unweighted case. Section \ref{sec:formulate} formulates the general sampling problem, and Section \ref{sec:sampling} introduces the novel MCMC method for sampling the null model. Section \ref{sec:stochastic_stability} considers the stochastic stability of the proposed sampler, and Section \ref{sec:wg_sims} performs an extensive simulation study to test the performance of the method against competing alternatives. Finally, we conclude in Section \ref{sec:wg_discussion}.

\section{Terminology}
\label{sec:wg_term}

This article is only concerned with directed graphs. Occasionally we consider unweighted graphs, which are denoted $G := (N, A)$ where $N := \{1,\ldots,n\}$ is a set of nodes and $A = (a_{uv})$ is the adjacency matrix. For weighted graphs the binary adjacency matrix is replaced with a weight matrix. Formally, $G := (N, W)$, where $W := (w_{uv})$ and $w_{uv} \in [0,\infty)$. The topology is implicit: $a_{uv} = 0$ if and only if $w_{uv} = 0$, or alternatively, $a_{uv} = w_{uv}^0$ with the convention that $0^0=0$. 

Define a node's \textit{out-degree} and \textit{in-degree} by $\degr^{-}_u := \sum_{v} a_{uv}$ and $\degr^{+}_u := \sum_{v} a_{vu}$ respectively, and collect them into vectors $\degr^- := (\degr^-_1,\ldots, \degr^-_n)^t$ and $\degr^+ := (\degr^+_1,\ldots, \degr^+_n)^t$. For weighted graphs, we define the node's out- and in-strength by $\str^{-}_u := \sum_{v} w_{uv}$ and $\str^{+}_u := \sum_{v} w_{vu}$, which are also collected into vectors $\str^- := (\str^-_1,\ldots, \str^-_n)^t$ and $\str^+ := (\str^+_1,\ldots, \str^+_n)^t$. If the graph to which an object belongs is unclear, we explicitly denote its dependence on the graph. For example we might write $W(G)$ instead of $W$.

\section{Motivating a Null Model for Weighted Graphs}
\label{sec:wg_motiv}
Null models have long been used to detect statistically significant patterns in unweighted networks. Such models have found application in a number of diverse fields, including sociology, ecology, categorical data analysis, systems biology, and community detection. While there exists an extensive literature for the unweighted case, very little has been developed for both defining and sampling an equivalent null model for weighted graphs. We now review the development of null models for unweighted graphs. 

\subsection{Null Models for Unweighted Graphs}
\label{sec:unweighted_motiv}

We define a family of distributions on the space $\mathcal{G}$ of unweighted graphs with $n$ nodes. Formally, let 
\begin{equation} \label{eq:holl}
    \Prob_{\theta}(G) := \kappa(\theta)^{-1}\exp{\left( \alpha^t
    \degr^- + \beta^t \degr^+ \right)},
\end{equation}
where $\kappa(\theta)$ is a normalizing constant and $\theta = (\alpha,\beta)^t$. The degree vectors are the sufficient statistics, or energies, of the distribution. The parameters $\alpha := (\alpha_1, \ldots, \alpha_n)^t$ and $\beta := (\beta_1,\ldots, \beta_n)^t$ control the distribution of out-degrees and in-degrees, with $\alpha_u$ and $\beta_u$ representing the sociability and popularity of node $u$. This is an exponential model, and may be viewed as a directed analogue of the $\beta$-model \citep{chatterjee_random_2011}, or as a special case of the $p_1$-family \citep{holland_leinghardt_1981}, whereby the reciprocity parameters are uniformly taken to be $\rho_{uv} = 0$. These $p_1$ models were introduced in the context of social network analysis, and were extended to the class of Markov Graphs by \citet{frank_markov_1986}, and eventually to the class of $p^*$, or exponential random graph models (ERGMs) \citep{wasserman_logit_1996}.

The model, and its undirected equivalent, are routinely used to measure the significance of properties observed in real-world networks. Measuring significance is useful a number of tasks; including for use in hypothesis testing, which is used to find evidence of local graph patterns \citep{milo_2002}. An example of such a pattern is reciprocity \citep{holland_leinghardt_1981}, which is often evident in social networks. Significance can also be optimised directly by including it in an objective function. This approach helps to discover network patterns, and is widely used for community detection \citep{newman_2004}.

In general, a practitioner will measure a property of interest in a network, which may be community structure, clustering, or a network motif, for example. This is usually summarised by a statistic $T:\mathcal{G} \to \mathbb{R}$, with large $T$ implying greater prevalence of the property. The observed value $t_0$ can only be interpreted in the context of \textit{the distribution of $T$ under a suitable null model}. To put this in a formal framework, we embed \eqref{eq:holl} in a larger exponential family
\begin{equation*}
    \Prob_{(\theta,\delta)}(G) := \kappa(\theta,\delta)^{-1} \exp{\left(\alpha^t\degr^- + \beta^t\degr^+ + \delta T(G)\right)},
\end{equation*}
which includes the statistic of interest. The null hypothesis that \eqref{eq:holl} provides a good fit of the network, i.e. that $t_0$ is not significant, is equivalent to testing $\delta = 0$ against the alternative $\delta \neq 0$. This is the approach suggested in \citet{holland_leinghardt_1981} to test the goodness of fit of the $p_1$-model, but can equally be interpreted as quantifying the extent to which $t_0$ is surprising under \eqref{eq:holl}. 

Notice that the hypothesis $\delta = 0$ is composite because the null depends on the unknown nuisance parameters $\theta$. The typical way to deal with this is to \textit{condition on the sufficient statistics}, which in this case are $\degr := ({\degr^-}^t,{\degr^+}^t)^t$. It is shown in \citet{lehmann_testing_2006} that tests based on this conditional distribution are optimal, i.e. the uniformly most powerful unbiased (UMPU) test of $\delta = 0$ against $\delta \neq 0$. If in fact the observed graph $G_0 \sim \Prob_{\theta_0}$ for some $\theta_0$, then the conditional distribution of $G_0$ given degrees is uniform on
\begin{equation*}
    \mathcal{G}(d_0) := \{G \in \mathcal{G} : \degr(G) = d_0\},
\end{equation*}
where $d_0 := \degr(G_0)$. This is the set of all graphs with the same degree sequence as $G_0$. This fact is obvious because \eqref{eq:holl} only depends on $G$ through the degrees.

In general, the conditional distribution of the test function is not available analytically, and so we typically resort to drawing samples $G_1,\ldots, G_N \sim  \text{Uniform}(\mathcal{G}(d_0))$. Significance $p \in [0,1]$ is then computed by comparing $t_0$ to the associated empirical distribution, i.e.
\begin{equation} \label{eq:significance}
    p := \frac{1}{N}\sum_{i=1}^N \mathbbm{1}_{[t_0,\infty)}(T(G_i)).
\end{equation}
The algorithms used to sample $G_1, \ldots G_N$ depend on our initial assumptions on the graph space. If $\mathcal{G}$ permits non-simple graphs, i.e. allows both self-loops and multiple edges, then it is straightforward to draw independent samples using the \textit{pairing model} (also known as the configuration model), which was first discussed by \citet{bollobas_probabilistic_1980, bender_asymptotic_1978}. However, in practice most networks are simple and if we restrict $\mathcal{G}$ accordingly, the situation becomes more complex. In particular, there is no straightforward method for drawing independent and exactly uniform samples. A common approach is to construct a Markov chain based on randomly rewiring edges, such that node degrees are exactly maintained \citep{ryser_1963, hakimi_realizability_1962, rao_1996}. This yields correlated samples which are asymptotically uniform, and can be treated as approximately independent if the chain is thinned appropriately. An alternative approach is to construct samples using sequential importance sampling \citep{bayati_2010, chen_conditional_2007, snijders_1991, blitzstein_diaconis_2011, zhang_chen_2013}.

\subsubsection{Why Conserve Degrees?}
\label{sec:whydegrees}

The most obvious starting point for a null model would be the directed Erd\H{o}s-Rényi model. However, this implies that node degrees are i.i.d. Binomial, and in particular that all nodes have the same expected degrees. In practice, degree distributions are rarely binomial, and are instead often heavy-tailed. This is a problem because the prevalence of many graph structures is tied to heterogeneity between nodes, and in particular the degree distribution. Practitioners are typically not interested in structure that arises purely as an artefact of this, and are instead looking for evidence of higher-order processes governing the formation of the network. Since the Erd\H{o}s-Rényi model cannot faithfully model degree distributions, it does not provide an adequate baseline with which to compare real networks to. By including $\degr^-$ and $\degr^+$ as sufficient statistics in \eqref{eq:holl}, the parameters $\alpha$ and $\beta$ can explicitly account for nodal heterogeneity, making $\Prob_{\theta}$ more suitable as a null model.

\subsection{Extending to Weighted Graphs}
\label{sec:weighted_motiv}

In the weighted case, a natural question is whether the strength sequences could substitute for the degrees in \eqref{eq:holl}. This approach has been proposed in the statistical mechanics literature, and is often referred to as the \textit{weighted configuration model} \citep{squartini_randomizing_2011, serrano_weighted_2005, serrano_correlations_2006}. It fails to faithfully model the topology of real networks. When $w_{uv}$ is continuous, all mass is on complete networks. When integer-valued, the probability of each edge existing approaches one for most real networks. The upshot is that degrees are important for conveying a graph's topology, and should be used \textit{in addition} to the strengths. Therefore we consider
\begin{equation} \label{eq:weighted_ergm}
    \Prob_{\theta}(G) := \kappa(\theta)^{-1}\exp{\left( \alpha^t
    \degr^- + \beta^t \degr^+ + \phi^t
    \str^- + \psi^t \str^+ \right)},
\end{equation}
which is an exponential family and an extension of \eqref{eq:holl}. $\kappa(\theta)$ denotes the normalizing constant and $\theta := (\alpha, \beta, \phi, \psi)^t$.  We assume that $\phi_u < 0$ and $\psi_u < 0$ for all $u \in N$ for reasons that will soon be clear. This model has appeared in \citet{mastrandrea_enhanced_2014}, where it was employed to reconstruct networks from node-level data. Following these authors, we refer to it as the \textit{directed enchanced configuration model} (DECM).

Clearly this model is neither elegant nor parsimonious. For any given node, there is likely to be high correlation between its fixed effects, and one naturally wonders whether the many parameters could be reduced by, for example, positing a simple functional relationship between degrees and strengths. The point, however, is that the model is general; by including fixed effects for both strengths and degrees it contains as sub-models many reasonable processes governing nodal-heterogeneity. This generality is essential for controlling for nodal effects when testing for higher-order processes that might explain network formation.

Unlike most exponential random graph models, the model is tractable and has a simple edge-level interpretation. An edge exists with probability
\begin{equation} \label{eq:ppos}
    P\{w_{uv} > 0\} = \frac{e^{\alpha_u+\beta_v}}{e^{\alpha_u+\beta_v} + \lambda_{uv}},
\end{equation}
where $\lambda_{uv} := - \phi_u - \psi_v$. A link is more likely to form if $u$ and $v$ are (topologically) sociable and popular respectively. The probability increases with $\phi_u$ and $\psi_v$, showing that edge formation also depends on the strength parameters.

Conditional on the edge $uv$ existing, its weight $w_{uv}$ follows an exponential distribution with rate $\lambda_{uv}$. The constraints on $\phi$ and $\psi$ ensure that this is positive. The exponential distribution is \textit{memoryless}, and so the probability of reinforcing an existing link by one unit is
\begin{equation*}
    P\{w_{uv} \geq x+1 \mid w_{uv} \geq x\} = e^{\phi_u + \psi_v},
\end{equation*}
for all $x > 0$. Since this is invalid when a link does not exist (i.e. when $x = 0$) there is a different cost for reinforcing an edge as opposed to forming a new edge. This permits network sparsity, and makes the model more suitable for modeling real networks than the weighted configuration model.

As in the unweighted case, the task is to use \eqref{eq:weighted_ergm} as a null model for quantifying the significance of a property of interest, which is measured by a statistic $T:\mathcal{G} \to \mathbb{R}$. Heuristic approaches have been proposed for this using maximum likelihood estimation \citep{mastrandrea_enhanced_2014, gabrielli_grand_2019}. Since \eqref{eq:weighted_ergm} is an exponential family, the MLE $\hat{\theta}$ can be found numerically as the solution to the $4n$ coupled equations given by setting observed sufficient statistics to their expectation. It is then straightforward to generate independent samples from the model with parameter $\hat{\theta}$. The observed statistic $t_0$ can then be compared to the sampled networks.

The aforementioned approach does not consider uncertainty around the MLE. In analogy to Section \ref{sec:unweighted_motiv}, a more formal approach considers the extended model
\begin{equation}
    \Prob_{\theta, \delta}(G) := \kappa(\theta, \delta)^{-1}\exp{\left( \alpha^t
    \degr^- + \beta^t \degr^+ + \phi^t
    \str^- + \psi^t \str^+ + \delta T(G) \right)},
\end{equation}
and formulates the problem as assessing the probability of observing $t_0$ given that $\delta = 0$. Within the likelihood framework, one approach is to appeal to Wilks' theorem, which states that the likelihood ratio statistic is, under regularity conditions, asymptotically Chi-squared with one degree of freedom. Unfortunately the conditions required to apply Wilks' theorem, and indeed even for appealing to the asymptotic consistency of the MLEs, do not hold in this model. The data $\{w_{uv}\}$ are not identically distributed under \eqref{eq:weighted_ergm}, and the number of independent parameters grow linearly with $n$. This observation has been made repeatedly for the unweighted case \eqref{eq:holl} \citep{holland_leinghardt_1981, snijders_1991, mcdonald_markov_2007}, but to our knowledge has received little attention in articles using \eqref{eq:weighted_ergm}.

Recall that the optimal test of $\delta = 0$ conditions on the sufficient statistics. The resulting null model would then be uniform on
\begin{equation*}
    \mathcal{G}(d_0,s_0) := \{G\in \mathcal{G} : \degr(G) = d_0, \str(G) = s_0\},
\end{equation*}
where $d_0 := \degr(G_0)$ and $s_0 := \str(G_0)$. This is the set of graphs conserving both degree and strength sequences exactly. Sampling uniformly from this set is in general a difficult problem, and we are not aware of any methods that have been proposed to achieve this. For this reason, our approach is to \textit{approximately condition} on the degrees, and consider instead the set
\begin{equation} \label{eq:targetspace}
    \mathcal{G}_m(d_0,s_0) := \{G\in \mathcal{G} : \|\degr(G) - d_0\|_{\infty} \leq m, \str(G) = s_0\},
\end{equation}
for $m > 0$. This maintains strengths exactly, and keeps all node degrees within $m$ of the observed values.

\section{The General Problem}
\label{sec:formulate}

We have motivated the task of sampling from the conditional distribution of \eqref{eq:weighted_ergm} given degrees and strengths. Indeed, this is the focus of the article. Nonetheless, other reasonable weighted null models exist. For example, \citet{palowitch_significance-based_2018} recently introduced the \textit{continuous configuration model}, which is a weighted extension the Chung-Lu model \citep{chung_average_2002, chung_connected_2002}. Since alternatives could be used, we keep the setting general.

The general problem is as follows. Let $\mathcal{G}$ be the space of graphs with $n$ nodes, which may prohibit edges in a set $\mathcal{F}\subseteq N^2$. That is, $G \in \mathcal{G}$ only if $w_{uv} = 0$ for all $uv \in \mathcal{F}$. This is typically employed to disallow self-loops, but can also be used to match any pattern of non-edges, including none at all. We hypothesise that the observed graph $G_0 \in \mathcal{G}$ is distributed according to some null model $P$. Viewing $P$ as defined on the space of weight matrices $[0,\infty)^{n\times n}$ and its Borel $\sigma$-algebra, the subset of this space not respecting $\mathcal{F}$ must be $P$-null. It is assumed throughout that $P$ has a density $f$ with respect to 
\begin{equation} \label{eq:Q}
\Lambda := \sum_{A \in \{0,1\}^{n\times n}} \lambda_{A},
\end{equation}
where $\lambda_A$ is $\|A\|_{0}$-dimensional Lebesgue measure on $\{ W \in [0,\infty)^{n\times n} : w_{uv}^0 = a_{uv}\}$, and where by convention $0^0 = 0$. These sets partition $[0,\infty)^{n\times n}$ and so $\{\lambda_A\}$ are mutually singular. The requirement ensures that if an edge exists, i.e. if $w_{uv}>0$, then it is continuous. It also permits network sparsity by allowing different topologies to have positive probability.

Let $d_0 := \degr(G_0)$ and $s_0 := \str(G_0)$. As mentioned, we are not able to sample from the set of graphs with degrees $d_0$ and strengths $s_0$, because the set is too constrained for our sampler to traverse. Instead, we opt to \textit{approximately} condition on the degrees. As we will see, this provides enough `slack' to construct a sampler. Define, for each integer $m > 0$, the function
\begin{equation} \label{eq:approxdegs}
    \degr_m(G) :=\mathbbm{1}_{N_m}(\degr(G)),
\end{equation}
where $N_m := \{d' : \|d' - d_0\|_{\infty} \leq m\}$ is a neighbourhood of $d$. Conditioning on this leads to graphs where each node has degrees that are within $m$ of the same node in $G_0$. For $m > n$, all graphs satisfy the degree condition and, in effect, we remove any conditioning on degree information. Fix some $m > 0$. The target distribution $\pi$ of our sampler is the conditional distribution of $P$ given the functions $\degr_m$ and $\str$. Its support is $\mathcal{G}_m(d_0,s_0)$. We now turn our attention to constructing a Markov chain capable of targeting this distribution.

\section{Randomising Weighted Graphs}
\label{sec:sampling}

Sampling from $\mathcal{G}_m(d_0,s_0)$ is difficult because the space is highly constrained. Here we develop a Markov chain approach to the problem. The algorithm is inspired by the rewiring chains that are already widely applied in the literature for unweighted graphs. It relies on repeatedly applying local moves, referred to as $k$-cycles \citep{gandy_2016}.

\subsection{Introducing $k$-cycles}
\label{sec:kcycles}

Consider the following `rewiring' update used to randomise simple unweighted directed graphs while preserving degrees exactly. Select two edges $u_1v_1$ and $u_2v_2$ uniformly at random. If $u_1$, $u_2$, $v_1$ and $v_2$ are not all distinct, or if either of $u_1v_2$ or $u_2 v_1$ are already edges, then reject and start again. Otherwise remove $u_1v_1$ and $u_2v_2$ from the edge set and replace them with $u_1v_2$ and $u_2v_1$. This local procedure is applied continually to randomise the network.

Here we introduce analogous updates for weighted graphs, referred to as $k$-cycles. These originally appeared in \citet{gandy_2016}. First fix two vectors of mutually disjoint nodes $(u_1,\ldots,u_k)$ and $(v_1,\ldots,v_k)$, where $k$ is between $2$ and $n$. A $k$-cycle attempts to update the weight matrix along the $2k$ coordinates
\begin{equation} \label{eq:coords}
    \{u_1v_1,u_1v_2,u_2v_2,\dots, u_kv_k, u_kv_1\},
\end{equation}
conditional on all other values. Figure \ref{fig:example_moves} depicts examples of these coordinates for different $k$. If $k = 2$ then four edges are potentially updated, which is similar to the rewiring move in unweighted graphs. It turns out, however, that we need to allow longer updates $k > 2$ to ensure irreducibility of the Markov chain.

\begin{figure}[tbp]
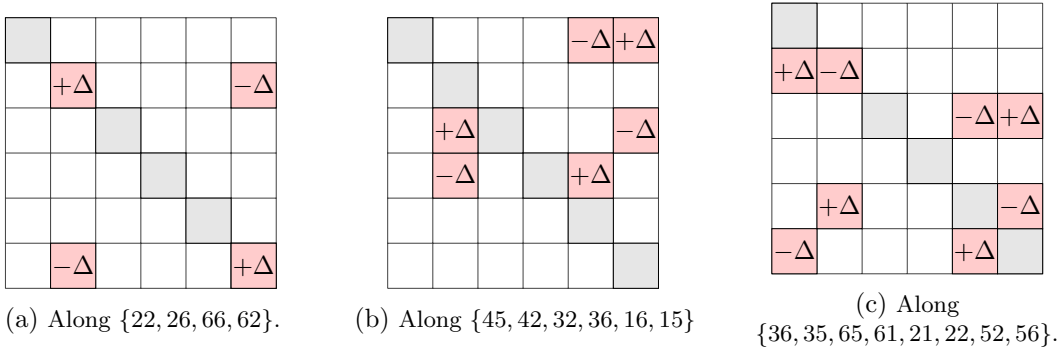

    \centering
    \includestandalone{figures/matrix_drawing}
    \caption{\small Example $k$-cycles of different lengths on a graph with 6 nodes. Self-loops are disallowed, as indicated by the gray squares. Therefore, the proposed cycles in Figures (a) and (c) would be rejected.}
    \label{fig:example_moves}
\end{figure}

Let $w := (w_1,w_2,\ldots,w_{2k})^t$ refer to the delineated edge weights corresponding to the coordinates \eqref{eq:coords}. Throughout this article, we refer to these weights as the \textit{cycle-weights}. These values must be updated so as to remain within $\mathcal{G}_m(d_0,s_0)$. In particular, conserving the strengths is equivalent to maintaining the marginals of the weight matrix. Because all edges outside of a $k$-cycle are considered fixed, conserving strengths is in fact equivalent to conserving the consecutive sums
\begin{equation} \label{eq:consum}
    (w_1+w_2,w_2+w_3,\ldots,w_{2k-1} + w_{2k}),
\end{equation}
exactly. Figure \ref{fig:example_moves} should help to convince the reader of this statement. We argue in Section \ref{sec:kcyclecond} that any update to $w$ must take the form $w + a \Delta$, where $a = (+1,-1,+1,\ldots,-1)^t$ and where the scalar $\Delta$ lies within a bounded interval that we are yet to define. This fact is also visualised in Figure \ref{fig:example_moves}.

\subsubsection{Relationship to Rewiring Moves}

For intuition, we clarify the relationship between rewiring moves in unweighted graphs and $2$-cycles. First assume that the weight matrix of an unweighted graph is synonymous with its adjacency matrix. A $2$-cycle would update the coordinates $\{u_1v_1, u_1v_2, u_2v_2, u_2v_1\}$. If in fact $w = (1,0,1,0)^t$, then setting $\Delta = -1$ leads to a new value $(0,1,0,1)^t$, performing the same edge replacement as a rewiring move. $\Delta = 0$ corresponds to rejecting a move, and would occur if say $w = (1,1,1,0)^t$. In unweighted graphs $\Delta$ could only ever lie in $\{-1,0,1\}$, rather than within a real interval as in the weighted case.

\subsection{Conditional Distribution along a $k$-cycle}
\label{sec:kcyclecond}

So far we have characterised $k$-cycles as updating certain subsets of the weight matrix while keeping all other entries fixed. Their definition is incomplete, as we are yet to describe how to make them reversible with respect to the target $\pi$. Since $k$-cycles are block updates of the weight matrix, it suffices for them to be reversible with respect to the \textit{full conditionals} of the cycle-weights. Here we derive these full conditionals. 

Formally, fix an arbitrary $k$-cycle and let $Q$ be the conditional distribution of $P$ given all weights outside of the $k$-cycle. $Q$ is taken to live on $(\Omega_{2k}, \mathcal{B}(\Omega_{2k}))$, where $\Omega_{d} := [0,\infty)^{d}$ is the $d$-dimensional non-negative orthant. Recall that $P$ has density $f$ with respect to \eqref{eq:Q}. Throughout this section, we let $f(x)$ denote $f$ evaluated at the weight matrix implied by letting the cycle-weights have value $x \in \Omega_{2k}$.

For intuition, we give an example of $Q$. Suppose $P$ is the DECM, then 
\begin{equation*}
    Q := \prod_{uv} \left(\frac{e^{\alpha_u+\beta_v}}{e^{\alpha_u+\beta_v} + \lambda_{uv}}Q_{uv} + \frac{\lambda_{uv}}{e^{\alpha_u+\beta_v} + \lambda_{uv}}\delta_{\{0\}}\right),
\end{equation*}
where the product is over the coordinates \eqref{eq:coords}, and notation is as in Section \ref{sec:weighted_motiv}. $Q_{uv}$ is the exponential distribution with rate $\lambda_{uv}$, and $\delta_{\{0\}}$ is the Dirac measure at zero. This is not, of course, the full conditional distribution of interest as we have neither conditioned on the strengths nor the degrees. Our approach is to first consider the strengths (Section \ref{sec:condonstrengths}), and then the degrees (Section \ref{sec:condondegrees}). 

\subsubsection{Conditioning on node strengths}
\label{sec:condonstrengths}

As argued in Section \ref{sec:kcycles}, conditioning $Q$ on the strengths is the same as conditioning on the consecutive sum \eqref{eq:consum}. This sum is formalised as a statistic $T: \Omega_{2k} \to \Omega_{2k-1}$. Rigorously proving the conditional distribution of $Q$ given $T$ is a difficult task, because $Q$ is neither fully discrete or continuous. This motivates a general definition of the conditional distribution, which is as follows.

\begin{definition}[Conditional Distribution] \label{def:condProb}
    A family $\mathcal{Q} := \{ Q_t : t \in \Omega_{2k-1}\}$ of probability measures on $\mathcal{B}(\Omega_{2k})$ is the conditional probability distribution of $Q$ given $T$ if 
    
    \begin{enumerate}
        \item $Q_t \{ T \neq t\} = 0$ for $TQ$-almost all $t$ in $\Omega_{2k-1}$, and
        \item if $g:\Omega_{2k} \to \mathbb{R}$ is nonnegative and measurable then $t \mapsto \int g(x) Q_t(\diff x)$ is measurable and
        \begin{equation} \label{eq:disintegration}
            \int g(x) Q(\diff x) = \int \int g(x) Q_t(\diff x) TQ (\diff t).
        \end{equation}
    \end{enumerate}
\end{definition}

First consider the level sets $\{T = t\}$ on which each $Q_t$ lives. It is easy to see that $\{T = t\}$ is a closed line segment $L_t$ that can be parameterised by
\begin{equation}
    \label{eq:lineSegmentParam}
    L_t(\Delta) := x + \Delta a,
\end{equation}
where $x$ is an arbitrary element of $L_t$ and $a := (+1,-1,+1,\ldots,-1)^t$ is the alternating vector described in Section \ref{sec:kcycles}. The scalar parameter $\Delta$ must lie in $[\Delta_{l}, \Delta_{u}]$ where $-\Delta_l$ and $\Delta_u$ are the smallest odd and even elements of $x$ respectively. The boundary points are $x_l := L_t(\Delta_l)$ and $x_u := L_t(\Delta_u)$.

Proposition \ref{prop:condProb} states the conditional distribution. The proof is provided in Appendix \ref{sec:condProbProof}. The proposition defines each $Q_t$ in terms of another distribution $\mu_t$ on $(\Omega_{2k}, \mathcal{B}(\Omega_{2k}))$, whose support is $L_t$. This is defined through
\begin{equation*}
\label{eq:line_measure}
    \mu_t(B) := \frac{\int_{L_t \cap B} f(x) ds}{\int_{L_t} f(x) ds},
\end{equation*}
for each $B \in \mathcal{B}(\Omega_{2k})$. Both integrals in this expression are line integrals, and the denominator serves as a normalising constant. 

\begin{proposition} \label{prop:condProb}
Fix any $t \in \Omega_{2k-1}$. If the boundary points satisfy $\|x_l\|_0 = \|x_u\|_0 = 2k-1$ then let
\begin{equation} \label{eq:2kminus1}
    Q_t := \frac{1}{\kappa_t}\left(f(x_l)\delta_{\{x_l\}} +  f(x_u)\delta_{\{x_u\}} + \alpha_t \mu_t\right),
\end{equation}
where $\alpha_t = \frac{1}{\sqrt{2k}}\int_{L_t}f(x)ds$ and $\kappa_t :=  f(x_l) + f(x_u) + \alpha_t$. Otherwise let
\begin{equation} \label{eq:conddistless}
Q_t :=
    \begin{cases}
    \delta_{\{x_l\}} & \text{if } \|x_l\|_{0} < \|x_u\|_{0} \\
    \delta_{\{x_u\}} & \text{if } \|x_l\|_{0} > \|x_u\|_{0} \\
    \kappa_t^{-1}\left(f(x_l)\delta_{\{x_l\}} +  f(x_u)\delta_{\{x_u\}}\right) & \text{if } \|x_l\|_{0} = \|x_u\|_{0} < 2k-1, \\
    \end{cases}
\end{equation}
where $\kappa_t := f(x_l) + f(x_u)$. The collection $\mathcal{Q} := \{Q_t : t \in \Omega_{2k-1}\}$ is the conditional distribution of $Q$ given $T$.
\end{proposition}

\subsubsection{Approximate conditioning on degrees}
\label{sec:condondegrees}

Suppose now that we consider degrees \textit{in addition} to the strengths, i.e. we wish to approximately condition $Q_t$ on the degrees. To formalise this, first fix an arbitrary $x \in L_t$ and let $G_x$ refer to the graph obtained by letting the cycle-weights take the value $x$, whilst keeping all other weights fixed. We then condition on $x \mapsto \degr_m(G_x)$, which is a map from $L_t \to \{0,1\}$. This statistic depends implicitly on the topology of the graph outside of the $k$-cycle, which is of course fixed. We assume that this topology is such that
\begin{equation} \label{eq:condset}
    \{x \in L_t : \degr_m(G_x) = 1\},
\end{equation}
is non-empty. The set of graphs for which this is empty is of course $\pi$-negligible because the graphs cannot lie within  $\mathcal{G}_m(d_0,s_0)$. Therefore, this case can be safely ignored.

Conditioning $Q_t$ on $x \mapsto \degr_m(G_x)$ is equivalent to restricting it to \eqref{eq:condset}, i.e. the points at which the associated graph has degrees close enough to the target vector. These graphs can have one of at most three topologies. If $x_l \neq x_u$ then the topologies of $G_{x_l}$ and $G_{x_u}$ are different, because the zero elements of $x_l$ and $x_u$ are distinct. If $x_1$ and $x_2$ are both in the interior of $L_t$, then the topology of $G_{x_1}$ and $G_{x_2}$ are the same, because all entries in $x_1$ and $x_2$ are positive. This shows that conditioning may assign zero probability to either of the boundary points, or to the entire interior of $L_t$.

Special attention should be given to the case where \eqref{eq:condset} is $Q_t$-negligible. This would happen, for example, if $x_l$ has more zeros than $x_u$, but also $\degr_m(G_{x_l}) = 0$. Another possibility is that $x_l$ and $x_u$ have the same number of zeros, but more than one, and $\degr_m(G_{x_l}) = \degr_m(G_{x_u}) = 0$. In both cases, it is easy to see that all points in \eqref{eq:condset} must have ties among positive elements, which is a $\pi$-negligible event. Therefore, the conditional distribution can be defined arbitrarily in this case.

\subsubsection{Example: Conditional Distribution for the DECM}
\label{sec:condfornull}

Here we specialise to the DECM. The resulting conditional distribution will be easy to sample directly, providing a convenient way to perform $k$-cycles.

First suppose that $x_l$ and $x_u$ each have one zero entry. To compute the line integral appearing in the conditional $Q_t$, first observe that $L_t(\Delta + \diff \Delta) = L_t(\Delta) + a \diff \Delta$ for any $\Delta \in (\Delta_l,\Delta_u)$, and so
\begin{equation*}
    \diff s = \|L_t(\Delta + \diff \Delta) - L_t(\Delta)\|_2 = \sqrt{2k}\diff \Delta,
\end{equation*}
where $\diff s$ is the differential on $L_t$. Therefore
\begin{align*}
    \int_{L_t} f(x) \diff s & = \sqrt{2k} \int_{\Delta_l}^{\Delta_u} f(x + a\Delta) \diff \Delta \\
    & = \sqrt{2k} f(x + a\Delta^*) (\Delta_u - \Delta_l),
\end{align*}
where $\Delta^* \in (\Delta_l,\Delta_u)$. Here we have used that $f(x + a\Delta) = f(x + a\Delta^*)$ for all $\Delta \in (\Delta_l,\Delta_u)$. This is true because $f$ depends only on degrees and strengths, which are invariant over such $\Delta$. It is also easy to verify that
\begin{align}
    f(x_l) &= e^{-(\alpha_{u_1} + \beta_{v_1})} f(x + a\Delta^*) \label{eq:flow}\\
    f(x_u) &= e^{-(\alpha_{u_2} + \beta_{v_2})} f(x + a\Delta^*), \label{eq:fup}
\end{align}
where $u_1v_1$ and $u_2v_2$ are the edges corresponding to the zero weights in $x_l$ and $x_u$ respectively. Putting this together, \eqref{eq:2kminus1} reduces to 
\begin{equation*}
    Q_t := \frac{1}{\kappa_t}\left(e^{-(\alpha_{u_1} + \beta_{v_1})}\delta_{\{x_l\}} + e^{-(\alpha_{u_2} + \beta_{v_2})} \delta_{\{x_u\}} + (\Delta_u - \Delta_l) \mu_t\right),
\end{equation*}
where $\mu_t$ is the uniform distribution on $L_t$, and $\kappa_t = e^{-(\alpha_{u_1} + \beta_{v_1})} + e^{-(\alpha_{u_2} + \beta_{v_2})} + \Delta_u - \Delta_l$. $Q_t$ in the remaining cases (shown in \ref{eq:conddistless}) are found similarly. Of course, $Q_t$ is not the distribution of interest as we also need to approximately condition on the degrees. This is straightforward and consists of restricting $Q_t$ to the appropriate parts of $L_t$, as was outlined in Section \ref{sec:condondegrees}. Direct sampling from both $Q_t$ and $Q_t$ given approximate degrees is straightforward. For general densities, however, the line integral would typically need to be computed by numerical integration. Furthermore, direct sampling from $\mu_t$ may not be possible, and could require more sophisticated methods like rejection sampling. 

\subsection{Performing the $k$-cycle}
\label{sec:performing}

We are now ready to describe the full $k$-cycle for the DECM. $\pi$-invariance is automatically satisfied since we  sample directly from the full conditional. One remaining issue, however, is that the parameter vectors $\alpha$ and $\beta$ in \eqref{eq:weighted_ergm} are unobserved. The full conditional still depends on these because the degrees have not been conditioned on exactly. Nonetheless, their influence is small when conditioning on $\degr_m$ for small $m$. One option is to assume $\alpha = \beta = 0$ so that they do not appear in the distribution. Another option is to estimate them via maximum likelihood. Algorithm \ref{alg:kcycle} gives pseudo-code for the algorithm which uses some assumed values $\hat{\alpha}$ and $\hat{\beta}$.

\begin{figure}[tbp]
\removelatexerror
\begin{algorithm}[H]
\small
\caption{\small A $k$-cycle for the DECM.}
\label{alg:kcycle}
\KwIn{$G$, $z$, $\hat{\alpha}$ and $\hat{\beta}$\;}
\lIf{$z \cap  \mathcal{F} \neq \emptyset$}{\Return $G$}
$x \gets W_z(G)$\;
$\Delta_l \gets -\min_{i}{(x_{2i+1})}$ and $\Delta_u \gets \min_{i}(x_{2i})$\;
\lIf{$\Delta_u = \Delta_l = 0$}{\Return $G$}
Let $z_l$ and $z_u$ be edges corresponding to elements of $x + a \Delta_l$ and $x + a \Delta_u$ that are zero respectively\;
$n_l \gets |z_l|$ and $n_u \gets |z_u|$\;
$p_l \gets p_u \gets p_{\text{int}} \gets 0$\;
\lIf{$\degr_m(G_{x + a\Delta_l})=1$ and $n_l \geq n_u$}{$p_l \gets \prod_{uv\in z_l} e^{-\hat{\alpha}_u - \hat{\beta}_v}$}
\lIf{$\degr_m(G_{x + a\Delta_u})=1$ and $n_u \geq n_l$}{$p_u \gets \prod_{uv\in z_u} e^{-\hat{\alpha}_u - \hat{\beta}_v}$}
Let $\Delta^* \in (\Delta_l,\Delta_u)$\;
\lIf{$\degr_m(G_{x + a\Delta^*})=1$ and $n_u = n_l = 1$} {
    $p_{int} \gets \Delta_u - \Delta_l$
}
$p^* \gets p_l + p_u + p_{\text{int}}$\;
\lIf{$p^* = 0$}{\Return $G$}
$u \sim \text{Unif}[0,p^*]$\;
\lIf{$u < p_l$}{$\Delta \gets \Delta_l$}
\lElseIf{$u < p_l + p_u$}{$\Delta \gets \Delta_u$}
\lElse{$\Delta \sim \text{Unif}(\Delta_l,\Delta_u)$}
\Return $G_{x + a\Delta}$\;
\end{algorithm}
\footnotesize
NOTES: $G$ is the current state of the chain, while $z$ is a set of coordinates of the form \eqref{eq:coords}. $\hat{\alpha}$ and $\hat{\beta}$ are $n$-vectors. Here, $G_x$ refers to the graph obtained by assigning weights $x$ to edges along the $k$-cycle, i.e. edges in $z$.
\end{figure}
\subsection{Combining $k$-cycles}
\label{sec:comb}

This section introduces an auxiliary variable method of selecting $k$-cycles. The $k$-cycle chosen at each iteration \textit{depends on the current state of the chain}. This allows for better mixing in both sparse and dense graphs.

\subsubsection{Motivating kernel selection}

To motivate our method, first recall the rewiring moves (discussed in Section \ref{sec:kcycles}) used for randomising unweighted directed graphs. The move selects two edges $u_1v_1$ and $u_2v_2$ randomly and attempts to replace them with $u_1v_2$ and $u_2v_1$. This is only possible if both $u_1v_2$ and $u_2v_1$ are not already in the edge set. If the network is sparse then the edge replacement has a high probability of succeeding. If it were dense, however, similar performance could be achieved by instead selecting $u_1v_1$ and $u_2v_2$ from the set of non-edges, rather than edges.

Closely related to the rewiring chains is a random walk that operates directly on the graph's adjacency matrix. This is often referred to as a checkerboard swap or tetrad move \citep{artzy-randrup_generating_2005, stone_checkerboard_1990, verhelst_2008, rao_1996, diaconis_1995}. It selects a $2\times 2$ submatrix at random and attempts to modify it with either 
\begin{equation*}
\left(\begin{matrix} +1 & -1 \\ -1 & + 1 \end{matrix}\right) 
\quad \text{ or } \quad
\left(\begin{matrix} -1 & +1 \\ +1 & -1 \end{matrix}\right),
\end{equation*}
and rejects if the resulting adjacency matrix is invalid. When successful, the move performs the same update as a rewiring move. However, the nodes are in effect chosen randomly, and so the performance will be poor in both sparse and dense graphs as the rejection rate is prohibitively high.

These local moves can be seen as \textit{Markov kernels}, and the method of selecting them is referred to as \textit{kernel selection}. The above discussion highlights the impact that kernel selection has on the practical efficiency of the resulting chain. Such considerations are exacerbated in the context of $k$-cycles; Proposition \ref{prop:condProb} shows that for $P$-almost all graphs, a $k$-cycle is unable to propose a new graph if there is more than one zero weight along the cycle. 

The naive approach would be to first sample $k \in \{2,\ldots,n\}$ and then each of $(u_1,\ldots,u_k)$ and $(v_1,\ldots,v_k)$ uniformly from the node set $N$ without replacement. The $k$-cycle is then formed as in \eqref{eq:coords}. This is analogous to the checkerboard/tetrad moves previously discussed, and is the approach used in \citet{gandy_2016}. In sparse graphs, however, the chance of only one zero cycle-weight is small, and the sampler can be prohibitively slow. Our aim in this section is to define a better strategy.

\subsubsection{Cycle selection as an auxiliary variable}
\label{sec:wg_auxvar}

The method of selecting a $k$-cycle can be interpreted as an auxiliary variable. Formally, let $\mathcal{Z}$ be the index set of all possible $k$-cycles, which corresponds to the collection of all sets of the form \eqref{eq:coords}. Note that permutations of \eqref{eq:coords} are considered equivalent here. We want the selected cycle $Z \in \mathcal{Z}$ to depend on the current state of the chain $G \in \mathcal{G}_m(d_0,s_0)$. Therefore, we let $Z \sim q_G$ where $q_G$ is the \textit{state-dependent distribution} of $Z$.

Selecting $k$-cycles this way does not generally maintain $\pi$-invariance. For this we must extend the state space to include the selection variable, and consider the properties of the Markov chain on the joint space. Formally, define the product space $\mathcal{Z} \times \mathcal{G}_m(d_0,s_0)$. The iterated integrals
\begin{equation*}
    \tilde{\pi}(g) := \int \int g(z,G) q_G(\diff z) \pi (\diff G),
\end{equation*}
for all non-negative Borel-measurable $g$ define a distribution $\tilde{\pi}$ on the joint space. Starting from $(z,G)$, the extended chain first samples $z' \sim q_G$, and proceeds to update the edge weights along the $k$-cycle defined by $z'$. If the extended chain is $\tilde{\pi}$-invariant then the marginal chain on $\mathcal{G}_m(d_0,s_0)$ is $\pi$-invariant. To maintain $\tilde{\pi}$-invariance, we simply need to adjust the full conditional of the weights along a $k$-cycle to additionally condition on $Z$.

\subsubsection{An efficient selection strategy}

We prioritise selecting cycles that have some chance of moving the chain to a new state. The main limitation of naively selecting $k$-cycles is sparsity. We therefore select new nodes by constructing an `alternating' cycle of out-edges and in-edges. We require two neighbourhood sets associated with each node. These are
\begin{align*}
    N^-_G(u) &:= \{v \in N : a_{uv} = 1\}\\
    N^+_G(u) &:= \{v \in N : a_{vu} = 1 \},
\end{align*}
which are the out-neighbours and in-neighbours of $u$ respectively. We start by sampling $k$ according to some distribution $\Gamma$ on $\{2,\ldots,n\}$. This should be positive everywhere to improve the stochastic stability of the sampler. We then sample $u_1v_1$ uniformly from the set of all edges in the graph. Starting from $u_1$, the remaining nodes are sampled by alternately walking through the out-neighbours and in-neighbours of the previous node. The full strategy is shown in Algorithm \ref{alg:selectkcycle}.

If the algorithm terminates at line 10, the resulting $k$-cycle has at most one zero weight and no fixed edges. If instead it returns $\emptyset$, then the strategy has failed to select a $k$-cycle and the Markov chain remains at the current state. If $k$ is small in comparison to the size of the network, nodes generally have more than two edges, and the pattern of prohibited edges $\mathcal{F}$ is not particularly complex, then the chance of failing to find a $k$-cycle is small. 

\begin{figure}[tbp]
\begin{algorithm}[H] 
\small
\caption{\small $k$-cycle selection strategy.}
\label{alg:selectkcycle}
\KwIn{$G$\;}
$k \sim \Gamma(\{2,\ldots,n\})$\;
$u_1v_1 \sim \text{Unif}(\{uv:a_{uv}=1\})$\;
\For{$l=2$ to $k$} {
 \lIf{$d^-_{u_{l-1}}(G) \leq 1$}{\Return{$\emptyset$}}
 $v_l \sim \text{Unif}(N^-_{u_{l-1}}(G) \setminus \{v_{l-1}\})$\;
 \lIf{$d^+_{v_l}(G) \leq 1$}{\Return{$\emptyset$}}
 $u_l \sim \text{Unif}(N^+_{v_{l}}(G) \setminus \{u_{l-1}\})$\;
}
\lIf{$\exists i,j$ distinct such that either $u_i=u_j$ or $v_i=v_j$}{\Return{$\emptyset$}}
\lIf{$u_kv_1 \in \mathcal{F}$}{ \Return{$\emptyset$}}
\Return{$\{u_1v_1, u_1v_2, \ldots u_k v_k, u_kv_1\}$}
\end{algorithm}
\footnotesize
NOTES: $G$ is the current state of the chain. Line 9 only checks if $u_kv_1$ is a prohibited edge because all other edges have positive weights, which by assumption implies they are not prohibited (see Section \ref{sec:formulate}).
\end{figure}

\subsection{The Overall Sampler}

One iteration of the full sampler tries to select a $k$-cycle with Algorithm \ref{alg:selectkcycle}. If successful, it then samples cycle-weights from its full conditional. Recall, however, that this must be adjusted to also condition on $Z$, which is the cycle selection variable. 

To do this, let $z \in \mathcal{Z}$ be a cycle chosen by Algorithm \ref{alg:selectkcycle}, and recall the notation where $G_x$ refers to the graph obtained by allowing cycle-weights to take the value $x$. For such cycles $\Delta_u - \Delta_l > 0$. Let $\gamma_l$ and $\gamma_u$ be the probability of selecting $z$ from $G_{x_l}$ and $G_{x_u}$ \textit{relative} to the chance of selecting it from some graph $G^* := G_{x + a\Delta}$ satisfying $\Delta \in (\Delta_l,\Delta_u)$. Assume also that there are no positive ties along the cycle-weights (positive ties are $P$-null). Then by following Algorithm \ref{alg:selectkcycle}, one can deduce that
\begin{align}
    \gamma_l &= \frac{M}{M-1}\frac{(\degr^-_{u_{1}}(G^*) - 1)(\degr^+_{v_{1}}(G^*) - 1)}{\sum_{uv \in z}  (\degr^-_{u}(G^*) - 1)(\degr^+_{v}(G^*) - 1) } \label{eq:alphal}\\
    \gamma_u &= \frac{M}{M-1}\frac{(\degr^-_{u_{2}}(G^*) - 1)(\degr^+_{v_{2}}(G^*) - 1)}{\sum_{uv \in z}  (\degr^-_{u}(G^*) - 1)(\degr^+_{v}(G^*) - 1) } \label{eq:alphau},
\end{align}
where $M$ is the total number of edges in $G^*$, and $u_1v_1$ and $u_2v_2$ are the edges corresponding to the zero weights in $x_l$ and $x_u$ respectively. Conditioning on the cycle selection strategy simply requires adjusting the boundary probabilities by the factors $\gamma_l$ and $\gamma_u$. This is shown in Algorithm \ref{alg:sampler}, which presents the full sampler.

\begin{figure}[tbp]
\begin{algorithm}[H]
\small
\caption{\small One iteration of the complete sampler.}
\label{alg:sampler}
\KwIn{$G$, $\hat{\alpha}$, $\hat{\beta}$}
$z \gets$ output of Algorithm \ref{alg:selectkcycle} applied to $G$\;
\lIf{$z = \emptyset$}{\Return G}
$G' \gets$ output of Algorithm \ref{alg:kcycle} to $G$, $z$ $\hat{\alpha}$ and $\hat{\beta}$, but adding $p_l \gets \gamma_lp_l$ and $p_u \gets \gamma_up_u$ after line 9 \;
\Return{G'}
\end{algorithm}
\footnotesize
NOTES: Input are the same as in Algorithm \ref{alg:kcycle}. The adjustment to Algorithm \ref{alg:kcycle} in line 3 accounts for the state-dependent selection of the $k$-cycle.
\end{figure}

\section{Stochastic Stability}
\label{sec:stochastic_stability}

Here we attempt to provide conditions under which the chain we have introduced is \textit{ergodic}; i.e. that it admits a unique invariant distribution. All proofs are provided in the Appendix. Ergodicity justifies the use of Monte Carlo averages through Birkhoff's ergodic theorem, which states that if $\{G_l\}_{l=0}^{\infty}$ is a Markov chain with unique invariant distribution $\pi$, and $h$ is integrable, then
\begin{equation*}
    \frac{1}{L} \sum_{l=0}^{L-1} h(G_l) \to \mathbb{E}(h(G))
\end{equation*}
as $L \to \infty$, and where the expectation is taken under $\pi$. 

Throughout this section we fix degrees and strengths $(d, s)$ and some $m > 0$, and consider the chain designed to sample from $\mathcal{G}_m(d,s)$. Proving irreducibility (Definition \ref{def:irreducible}) in the general case is difficult, and we are only able to provide results for $m \geq n$, i.e. when there is in effect no conditioning on degrees. Further work is required to establish conditions for $m < n$. Nonetheless, our simulations appear to show that the sampler can traverse a large number of topologies even when $m = 1$, and is capable of rapidly reaching the mode of $\pi$ when the initial state is far in the tail of $\pi$. An example of this is provided in Section \ref{sec:mixing}.

Now assume that $m \geq n$. It turns out that the chain is not ergodic for all strength sequences. Nonetheless, ergodicity holds for strength sequences produced by $P$-\textit{almost all} graphs. To formalise this idea, let $\{U_i \times V_i \subseteq N^2 \}_{i \in I}$ be a collection of non-empty and distinct sets for which
\begin{equation}
                \sum_{u \in U_i} s^-_u = \sum_{v \in V_i} s^+_v,
\label{eq:in_out_match}
\end{equation} 
and such that there does not exist non-empty $U \times V \subset U_i \times V_i$ on which \eqref{eq:in_out_match} holds. Let $\tilde{\mathcal{G}}_m(d,s) \subseteq  \mathcal{G}_m(d,s)$ be the set of graphs for which $uv$ is an edge only if $uv \in U_i \times V_i$ for some $i \in I$. Ergodicity will hold for admissible strengths, as defined in Definition \ref{def:admissible}.
\begin{definition}[Admissible Strengths]
The strength vector $s$ is admissible if $\{U_i\}_{i \in I}$ and $\{V_i\}_{i \in I}$ each partition $N$ and $\tilde{\mathcal{G}}_m(d,s)$ is non-empty.
\label{def:admissible}
\end{definition}
Notice that \eqref{eq:in_out_match} is always satisfied for $U_i=V_i=N$. If these are the unique sets satisfying \eqref{eq:in_out_match}, then admissibility simply requires that the reference set is non-empty. 

Proposition \ref{prop:inadmissible_null} verifies that if the network is generated from some law absolutely continuous with respect to $P$, then the observed strengths will not be inadmissible. It also has implications for the topology of graphs that produced the strengths.
\begin{proposition}
The set of graphs producing inadmissible sequences is $P$-negligible, where $P$ is as defined in Section \ref{sec:formulate}. Moreover, for any admissible $s$ the set of graphs not in $\tilde{\mathcal{G}}_m(d,s)$ is $P$-negligible.
\label{prop:inadmissible_null}
\end{proposition}
Loosely speaking, ergodicity of a chain requires that it is irreducible, aperiodic and recurrent. By construction, the chain is aperiodic and has a unique invariant distribution, which will imply that it is recurrent. Therefore, the work is in demonstrating the property of irreducibility, defined as follows. 

\begin{definition}[$\varphi$-irreducibility]
A Markov chain on $(\mathcal{X}, \mathcal{B})$ with kernel $\Phi$ is $\varphi$-irreducible if there exists a measure $\varphi$ on $\mathcal{B}$ such that for all $x \in \mathcal{X}$ and $A$ for which $\varphi(A) > 0$, there exists some $n > 0$ satisfying $\Phi^n(x,A) > 0$.
\label{def:irreducible}
\end{definition}
Definition \ref{def:irreducible} shows that one can choose an arbitrary measure $\varphi$ when establishing irreducibility. If the property is satisfied then there exists a unique (up to null sets) `maximal' irreducibility measure $\psi$, in the sense that any other irreducible measure must be absolutely continuous with respect to $\psi$. For more details on this, see \citet[article 4]{meyn_tweedie_glynn_2009}. The next proposition ties irreducibility to admissibility of the strength sequence.
\begin{proposition}
If the strength vector $s$ is admissible then the resulting Markov chain is $\varphi$-irreducible.
\label{prop:irreducibilty}
\end{proposition}
Suppose $s$ is admissible. By construction, $\pi$ is an invariant distribution of the chain. Since the chain is also $\psi$-irreducible, it is recurrent \citep[Proposition 10.1.1]{meyn_tweedie_glynn_2009} and the invariant distribution is unique \citep[Proposition 10.4.4]{meyn_tweedie_glynn_2009}. This distribution is then the maximal irreducibility measure. This discussion is formalized in the following corollary.
\begin{corollary}
If $s$ is admissible then the resulting chain has $\pi$ as a unique invariant distribution. 
\end{corollary}

\section{Experiments}
\label{sec:wg_sims}
Here we assess the performance of the sampler introduced in Section \ref{sec:sampling}. The sampler was coded in C++, and all experiments were performed on an Intel Core i5 2GHz CPU. We first empirically analyse its efficiency in Section \ref{sec:mixing}. This will demonstrate its ability to randomise large networks. In Section \ref{sec:commdet}, we use the sampler as a null model for detecting patterns in weighted networks and compare its performance to competing methods.

\subsection{Efficiency of the Sampler}
\label{sec:mixing}

We use the sampler to randomise a large, sparse, and highly structured network. The randomisation maintains strengths exactly and keeps all node degrees within $\pm$ 1 of the initial network. The graph to be randomised has $n = 10^3$ nodes, $250$ of which are assigned as `core' nodes, and $750$ as `periphery' nodes. The core is partitioned into 5 cliques of 50 nodes, while the periphery is partitioned into 75 cliques of 10 nodes. The subgraph of each clique is complete; that is every node has a directed link to all other nodes in the community. All 80 clique are connected by a single bridge to the rest of the network. Specifically, we add 79 `bridge' links by creating a single link from the first clique to the second, a link from the second to the third, etc. Weights of all edges are sampled independently from the exponential distribution with mean $10^3$.

The adjacency matrix of this network is shown in the top left panel of Figure \ref{fig:sim1_results}. This initial state is far from the mode of the posterior, which is concentrated on matrices similar to that in the bottom right panel. The network is particularly difficult to randomise because creating links between cliques requires first choosing a $k$-cycle that includes a bridge link. Nonetheless, the sampler reached the network in the bottom right panel in under 35 seconds. We also attempted the randomisation without the auxiliary kernel selection method of Section \ref{sec:wg_auxvar}, instead selecting cycles as in \citet{gandy_2016}. However, this approach was unable to reach the mode within a reasonable time. This demonstrates the importance of the state-dependent kernel selection for the efficiency of the sampler.

Recall that in Section \ref{sec:stochastic_stability} we considered the irreducibility of the chain. We were, however, only able to obtain results for $m \geq n$, i.e. when there is no conditioning on the degrees. This experiment has almost exactly conditioned on the degrees, and shows that the sampler remains capable of rapidly randomising the network, and also of traversing different graph topologies. Although this certainly does not constitute a proof of irreducibility, it warrants further research in this direction.


\begin{figure}[tbp]
    \centering
    \begin{subfigure}{0.8\textwidth}
    \begin{subfigure}{0.48\textwidth}
        \includegraphics[width=\linewidth]{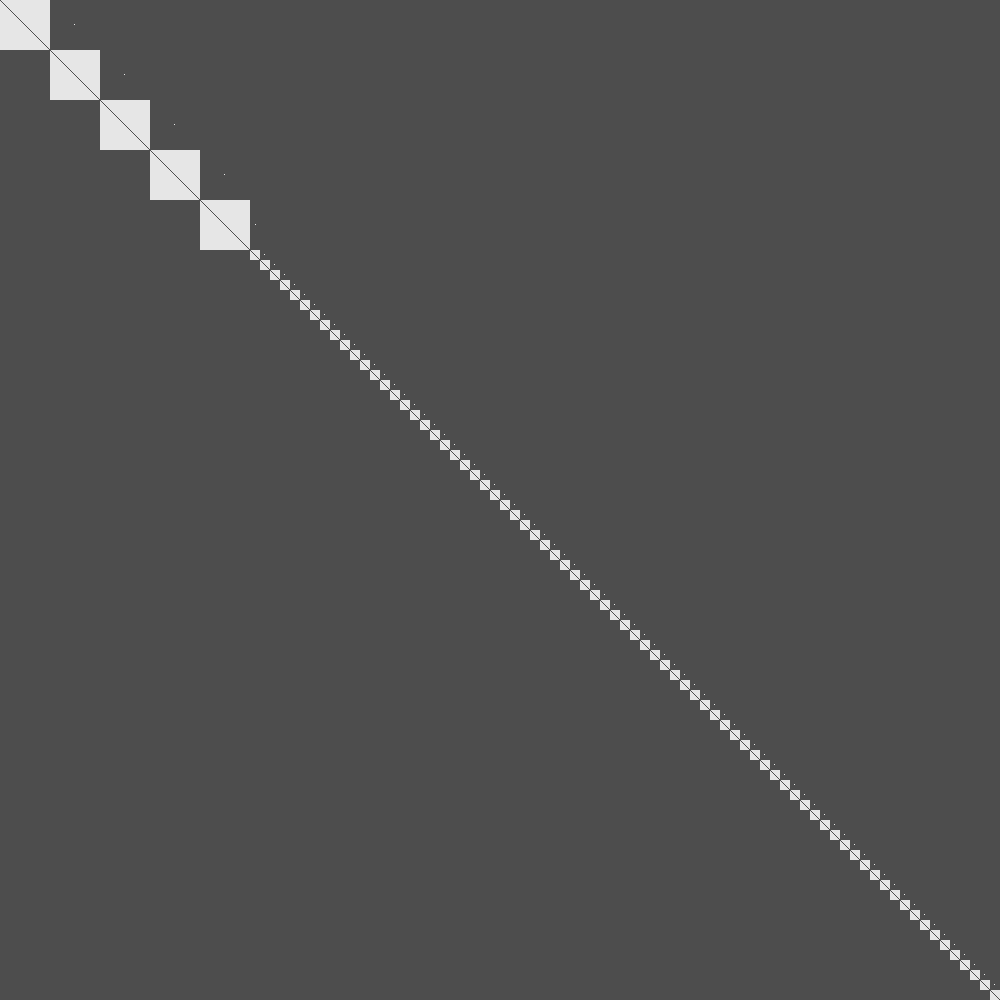}
        \caption{\small Initial Graph} \label{fig:sim1_0}
    \end{subfigure} \hspace*{\fill}
    \begin{subfigure}{0.48\textwidth}
        \includegraphics[width=\linewidth]{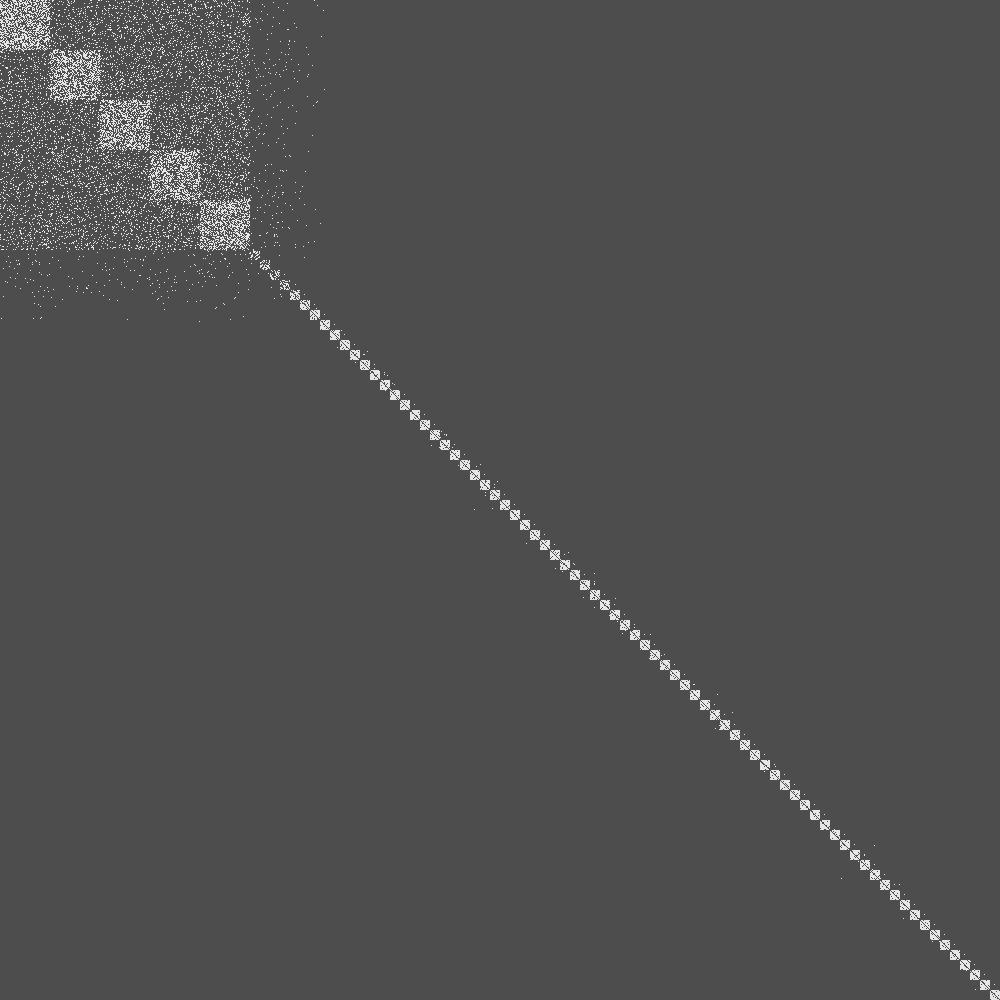}
        \caption{\small $10^6$ Iterations} \label{fig:sim1_1}
    \end{subfigure}
    \medskip
    \begin{subfigure}{0.48\textwidth}
        \includegraphics[width=\linewidth]{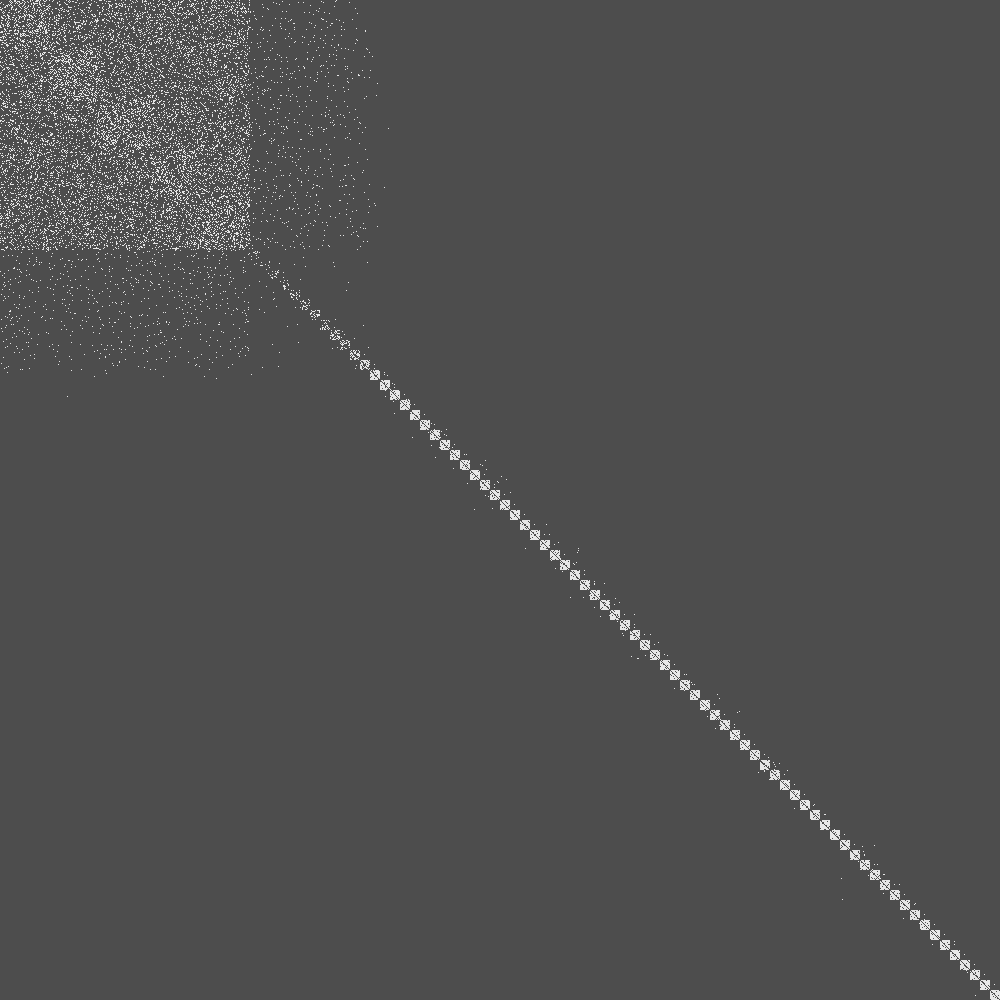}
        \caption{\small $2 \times 10^6$ Iterations} \label{fig:sim1_2}
    \end{subfigure} \hspace*{\fill}
    \begin{subfigure}{0.48\textwidth}
        \includegraphics[width=\linewidth]{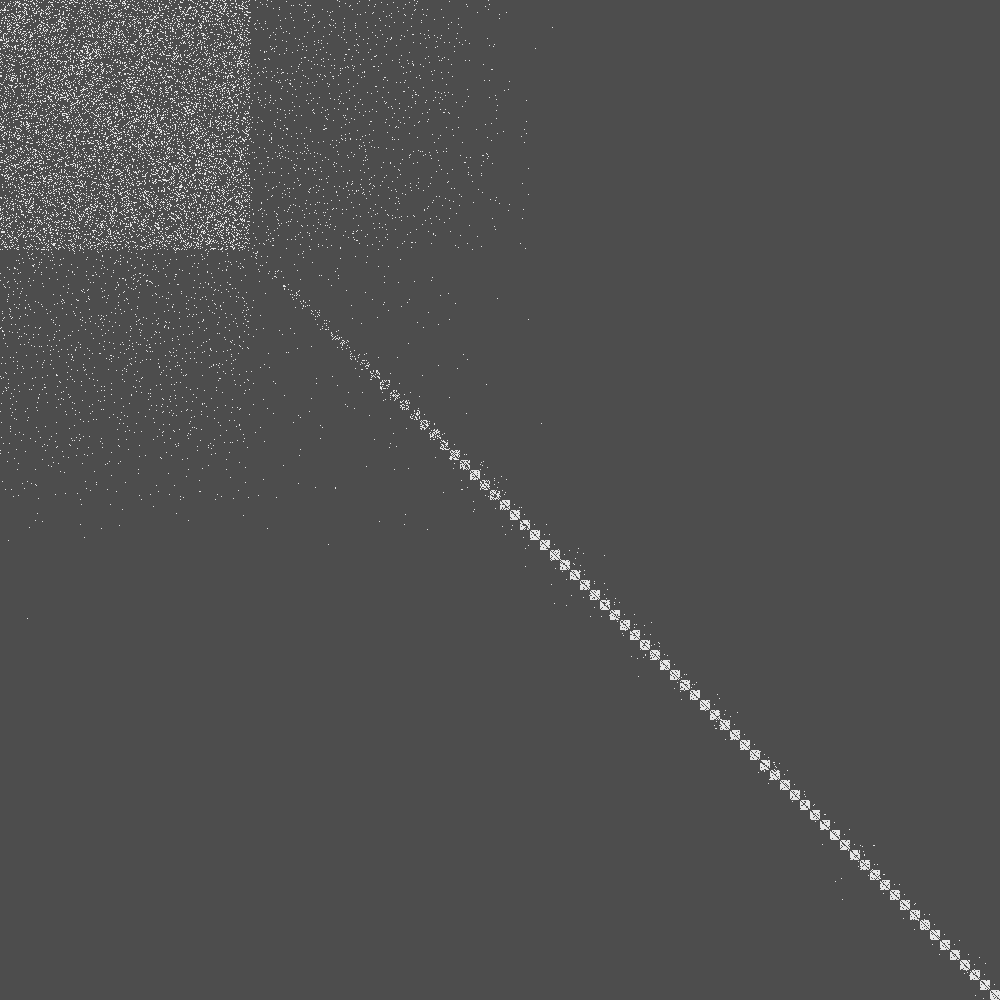}
        \caption{\small $3 \times 10^6$ Iterations} \label{fig:sim1_3}
    \end{subfigure}
    \medskip
    \begin{subfigure}{0.48\textwidth}
        \includegraphics[width=\linewidth]{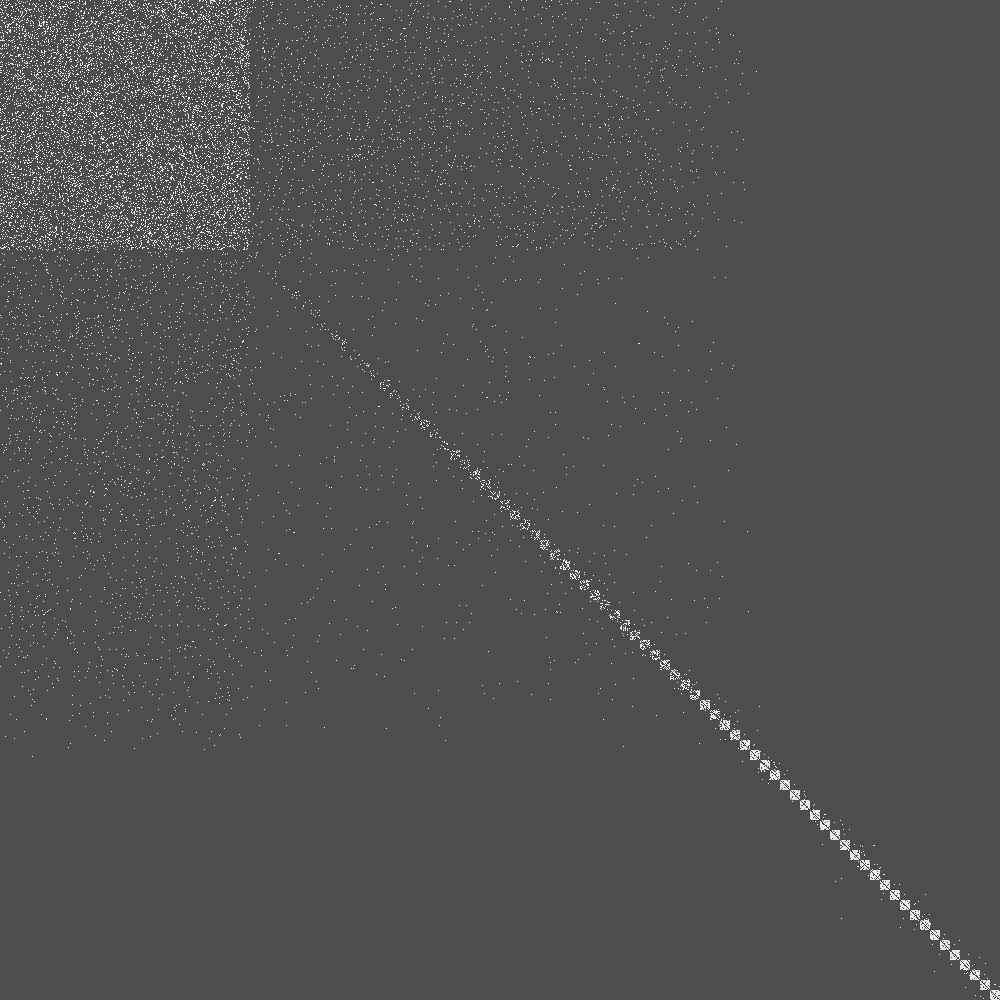}
        \caption{\small $4 \times 10^6$ Iterations} \label{fig:sim1_4}
    \end{subfigure} \hspace*{\fill}
    \begin{subfigure}{0.48\textwidth}
        \includegraphics[width=\linewidth]{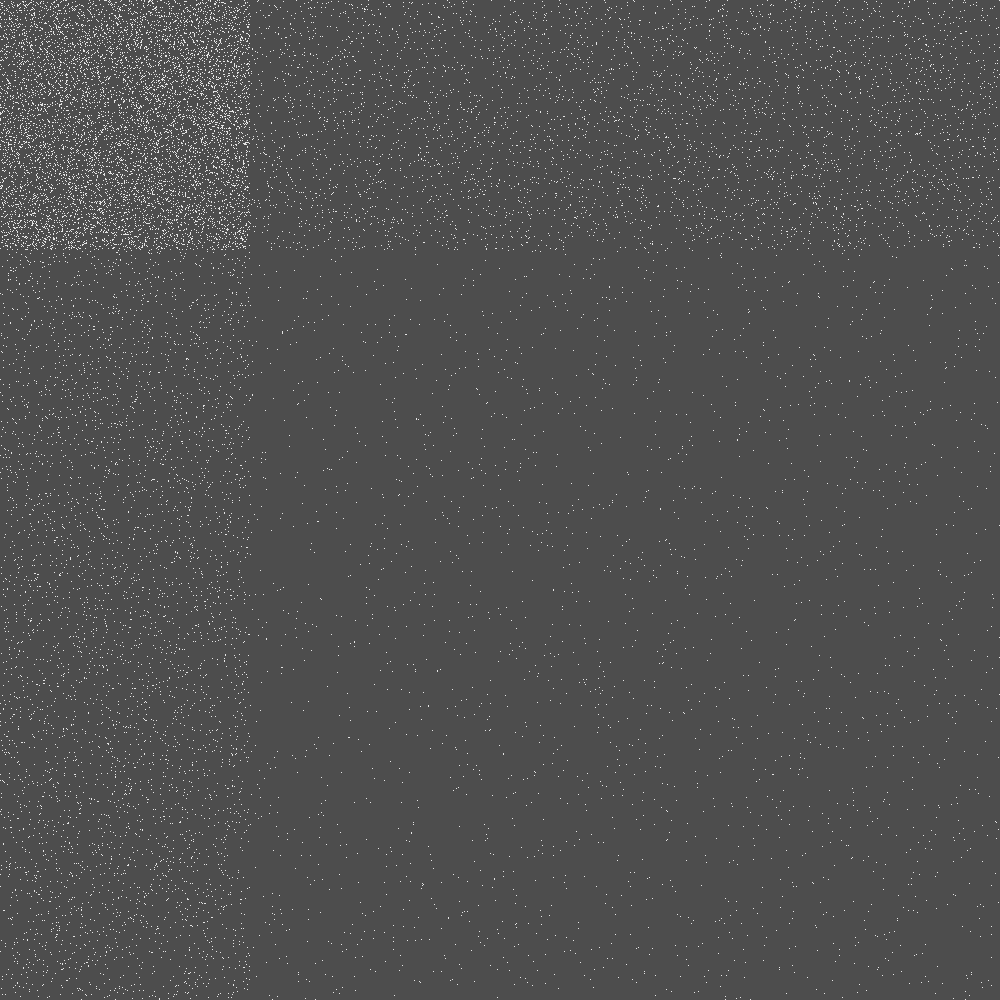}
        \caption{\small $15 \times 10^6$ Iterations} \label{fig:sim1_5}
    \end{subfigure}
    \end{subfigure}
    \caption{\small Adjacency matrices of networks at different stages of the randomisation.}
    \label{fig:sim1_results}
\end{figure}


\subsection{Significance of Community Structure in Benchmark Networks}
\label{sec:commdet}
This section uses the sampler to assess community structure in simulated networks, and compares the method's performance to alternative null models in a power study. The ground-truth community structure in the networks is known. Such `ground-truth' networks are usually simulated and are often termed benchmark models. Below, we introduce the benchmark model used and detail the parameters used for the simulation. We then describe the power study, competing methods and present the results.

\subsubsection{Degree and Strength Corrected Stochastic Block Model}

An early benchmark for unweighted and undirected graphs was suggested by \citet{girvan_community_2002}. Although simple, it does not account for heterogeneous community sizes and degrees. Modeling realistic degree distributions, which are heavy-tailed, is critical to the suitability of a benchmark. Heavy-tailed degrees can lead algorithms to group nodes with large degrees irrespective of their true memberships \citep{karrer_stochastic_2011}. \citet{lancichinetti_benchmark_2008} introduced the LFR benchmark, which overcomes these shortcomings. This was extended to weighted and directed graphs in \citet{lancichinetti_benchmarks_2009}. Here we use a benchmark model that accounts for heterogeneous group sizes, degrees and strengths. The benchmark is simple and similar in spirit to the \textit{weighted stochastic block model}s (WSBMs) proposed in \citet{aicher_learning_2015} and \citet{palowitch_significance-based_2018}. 

The model is a straightforward generalisation of the null model introduced in Section \ref{sec:weighted_motiv}, with an additional parameter to control tendency towards clustering. It is defined for $n > 1$ nodes and $K > 1$ possible community assignments. Recall from Section \ref{sec:weighted_motiv} the $n$-vectors $\alpha$, $\beta$, $\phi$ and $\psi$, where we constrained $\phi_u < 0$ and $\psi_u < 0$. Let $c$ be an $n$-vector representing a community partition of the nodes so that $c_u \in \{1,\cdots, K\}$.  The strength of community structure is controlled by a scalar parameter $\theta \geq 1$ that, roughly speaking, represents the relative edge formation probability (or average edge weights) for intra-community compared to inter-community links. Each potential edge $uv$ is associated with a factor $\theta_{c_uc_v}$, which is equal to $\theta$ if $c_u=c_v$, and is otherwise 1.

We describe the model at the edge-level, which will suggest a generative approach to drawing samples from it. As in Section \ref{sec:weighted_motiv}, edges are assumed conditionally independent given the parameters, and thus this description will fully define the likelihood for the network. The probability that an edge forms between two distinct nodes is
\begin{equation} \label{eq:probexist}
    P\{w_{uv} > 0\} =
    \min\left(\frac{e^{\alpha_u + \beta_v}}{e^{\alpha_u + \beta_v} + \lambda_{uv}}\theta_{c_uc_v},1\right),
\end{equation}
where $\lambda_{uv} = -\phi_u - \psi_v$. Conditional on edge existence, the weights are exponentially distributed with mean
\begin{equation*}
    \mathbb{E}[w_{uv} \mid w_{uv} > 0] = \frac{\theta_{c_uc_v}}{\lambda_{uv}}.
\end{equation*}

This model can be seen as a stochastic block model generalised to account for a wide range of degree and strength distributions. When $\theta=1$ it collapses to the null model of Section \ref{sec:weighted_motiv}. The model can be extended in multiple ways.  We have assumed no background nodes and no overlapping communities. For possible ways to extend in this direction, please see \citet{palowitch_significance-based_2018}. In addition, $\theta$ could be replaced with group-specific parameters.

\subsubsection{Simulation Parameters}
\label{sec:simparams}

The distribution of group sizes, degrees and strengths are chosen to reflect the heavy-tailed nature of these quantities in real networks. For this, we follow an approach that is close to \citet{lancichinetti_benchmarks_2009}. Formally, we iteratively draw group sizes from a discrete power law with exponent $\alpha_1$ truncated to between $s_{\text{min}}$ and $s_{\text{max}}$. Continue drawing new communities until the sum of all sizes is at least $n$, and then scale sizes proportionately until the total size is $n$. We then randomly assign nodes to the communities, such that the group sizes are respected.

We now describe all parameter values used in the simulations. The simulation requires applying our sampler thousands of times to different simulated networks. For this reasons, we consider only relatively small networks by letting $n=10^2$. For drawing group sizes, we let $\alpha_1 = 2$, $s_{\text{min}} = n / 5$ and $s_{\text{max}} = 3 s_{\text{min}} / 2$. The parameter $\theta$, which induces community structure, will be varied on a grid to assess the power of different methods at detecting deviations from the null.

\subsubsection{Competing Null Models}

The proposed method is compared to two alternative null models. The first is a weighted version of the Erd\H{o}s-Rényi model (WER). Let $G$ be a draw from the benchmark model, and $a_T := \sum_{u,v} a_{uv}(G)$ and $w_T := \sum_{u,v} w_{uv}(G)$ be the number of edges and total edge weights in $G$ respectively. WER draws independent and identically distributed edges according to 
\begin{equation*}
    P\{w_{uv} > 0\} = \frac{a_T}{n(n-1)},
\end{equation*}
with $u$ and $v$ distinct. Weights are then exponential with mean
\begin{equation*}
    \mathbb{E}(w_{uv} \mid w_{uv} > 0) = \frac{w_T}{a_T}.
\end{equation*}
The second model considered is the \textit{continuous configuration model} (CCM) introduced in \citet{palowitch_significance-based_2018}. This is a weighted extension of the Chung-Lu model \citep{chung_average_2002, chung_connected_2002} and unlike WER, has the advantage of matching the degrees and strengths of $G$ in expectation. Edges are formed independently with probability

\begin{equation*}
    p_{uv} := P\{w_{uv} > 0\} = \min\left(\frac{d_u^-(G) d_v^+(G)}{a_T},1\right),
\end{equation*}
and weights are exponential with mean
\begin{equation*}
    \mathbb{E}(w_{uv} \mid w_{uv} > 0) = \frac{s_u^-(G) s_v^+(G)}{w_T}\frac{1}{p_{uv}}.
\end{equation*}
CCM must permit self-loops else the degrees and strengths of $G$ are not correctly matched.

\subsubsection{The Power Study and Results}

The study is split into two phases. In both parts we are interested in assessing the power of the competing methods at correctly detecting community structure, where the level of such structure is controlled by $\theta$. This parameter will be varied from no clustering to levels where the clustering is quite apparent. Formally, we consider $\theta\in\{\theta_1,\ldots,\theta_{L}\}$ where $1 = \theta_1 < \ldots < \theta_{L}$. For each method and $\theta_l$, we repeatedly complete the following three steps.
\begin{itemize}
    \item Draw $G$ according to the benchmark model for $\theta_l$, and with all other parameters as described in Section \ref{sec:simparams}. Compute $t_0 = T(G)$, where $T$ is some statistic measuring the strength of clustering. 
    \item Draw samples $G^{(1)},\ldots,G^{(N)}$ using the method and let $t_1,\ldots,t_N$ be the associated test statistics.
    \item Compute the empirical significance ($p$-value) as in \eqref{eq:significance}.
\end{itemize}
This process is repeated $5 \times 10^3$ times in order to obtain a distribution over the significance statistics. If a method performs well then the $p$-values should be roughly uniformly when $\theta = 1$ and have high power for $\theta>1$.

The first phase of the study pretends that the true communities in the benchmark graphs are unknown, and applies a standard community detection algorithm to recover the structure. For this, we employ \textit{WalkTrap} \citep{pons_computing_2005}, however note that numerous alternatives could be used instead. The algorithm returns a graph partition, and we let $T$ be modularity computed on this partition. Figure \ref{fig:sim2_results} shows the comparative performance of different methods as $\theta$ is varied from $1$ to $2$ in increments of $0.2$. The figure shows that the proposed method outperforms the competing null models. In particular, when $\theta = 1$ the null model which conditions tightly on degrees ($\pm 1$) is close to uniform, as desired. This is not the case for either WER or CCM. Our method has a power advantage over both WER and CCM when the clustering effect is quite slight. This is expected: conditioning can act to improve relevance to the data at hand, and improve power against subtle alternatives. All methods perform well when $\theta$ is large.

The second phase looks to better understand the effect that approximate conditioning on the degrees has on the performance of the method. In order to illustrate this, we use a statistic that is deliberately sensitive to graph density. This is
\begin{equation*}
    T(G) := \sum_{u,v} a_{uv}(G) \delta(c_u,c_v),
\end{equation*}
where $\delta$ is the Kronecker delta function. This measures total within-community edges. In practice, this statistic would not be used because we have modularity, which explicitly measures clustering \textit{relative} to the configuration model and thus accounts for degree distributions. Nonetheless, it is not possible to make general graph statistics invariant to degree distributions, and so this example still has strong practical implications.

Figure \ref{fig:sim2_results} presents the results for the second phase. When $\theta = 1$, none of the null models are exactly uniform and there appears to be a bias towards high $p$-values. This is expected, as $T$ is highly sensitive to degrees. Nonetheless, when degrees are conditioned $\pm 1$ we get quite close to uniform because the effect of the unknown parameters (which were estimated by MLE) is minimised. Again, we see that approximate conditioning improves power against subtle alternatives.

\begin{figure}[ht]
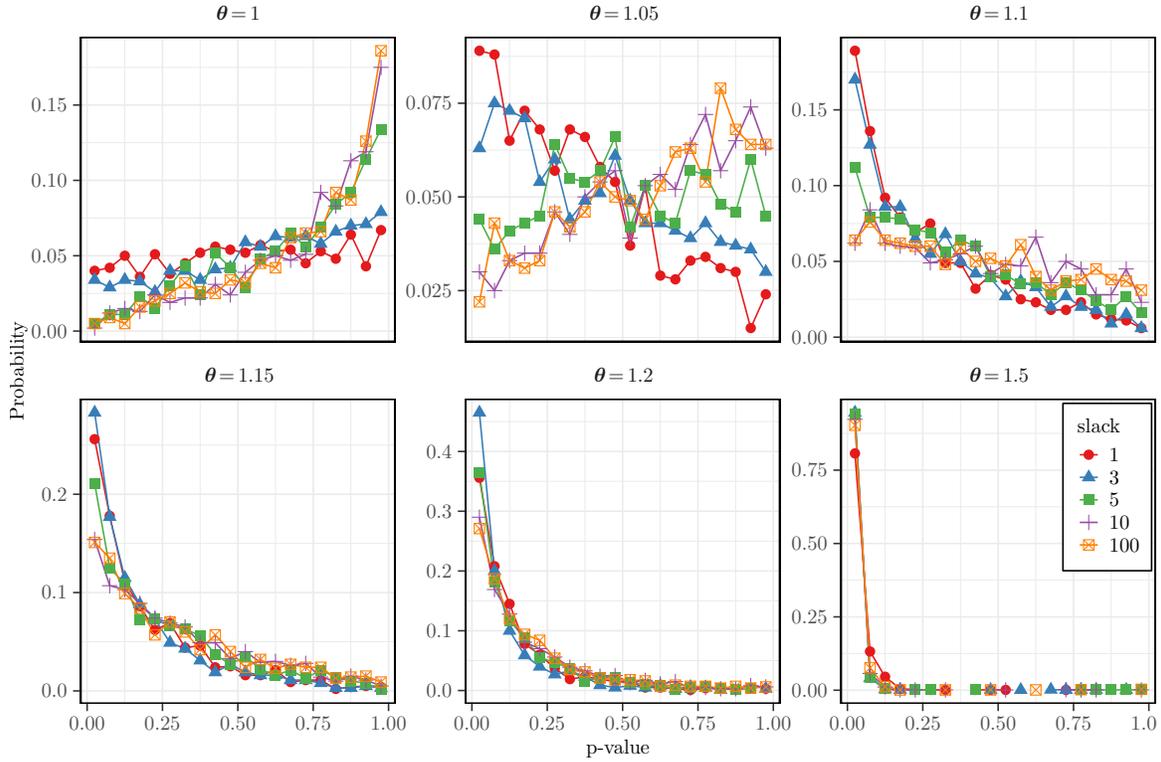
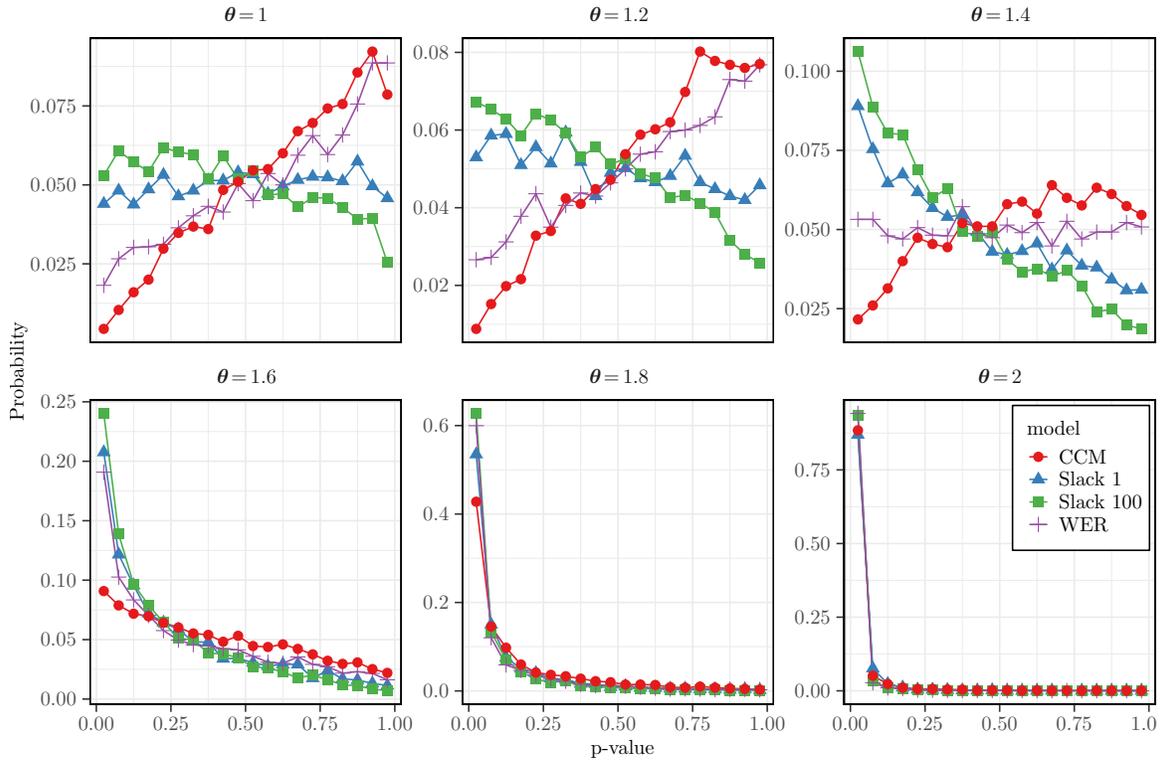

\centering
\begin{subfigure}{\linewidth}
    \centering
    \scalebox{1}{\includestandalone{figures/sim2a_results}}
    \caption{\small First phase.}
    \label{fig:sim2a_results}
    \end{subfigure}
\begin{subfigure}{\linewidth}
    \centering
    \scalebox{1}{\includestandalone{figures/sim2b_results}}
    \caption{\small Second phase.}
    \label{fig:sim2b_results}
\end{subfigure}
\caption{\small Comparative performance of different null models in the power study. \textit{Slack m} refers to the null model where node degrees are maintained $\pm m$. When $m=1$ degrees are almost exactly conditioned on. \textit{Slack 100} in effect performs no degree conditioning.}
\label{fig:sim2_results}
\end{figure}

\FloatBarrier

\section{Discussion}
\label{sec:wg_discussion}

This article has suggested a null model for weighted graphs. The model fixes node strengths and approximately fixes node degrees to within $\pm 1$ of the values of an observed network. It can be employed to assess the statistical significance of patterns observed in networks. We have proposed an MCMC sampler for drawing samples from the model, and have shown empirically that it is capable of sampling large and sparse networks. We performed an extensive power study to compare the performance of the null model to alternatives. The model compares favorably and appears capable of detecting subtle patterns, while also effectively controlling for nodal heterogeneity.

The work can be extended in a number of ways. We have only considered directed graphs, however the methods could in principle be extended to undirected graphs. From initial work in this direction, it appears that the undirected equivalent of $k$-cycles is not sufficient for maintaining irreducibility of the sampler. This challenge would need to be overcome. Another open question is whether the sampler is capable of reaching all possible graph topologies in the general case. Here, we only proved irreducibility when conditioning on strengths. Nonetheless, the simulation presented in Section \ref{sec:mixing} showed that the sampler is capable of traversing different graph topologies, even when conditioning tightly on degrees. The methods introduced in this article can also be used for Bayesian reconstruction of financial networks (see, for example, \citet{gandy_2016}). Since in sparse networks, our sampler is considerably more efficient than the original sampler in \citet{gandy_2016}, we expect the method to be useful in this field.

\bibliographystyle{apalike}
\bibliography{references}

\section{Appendix A: Proof of Proposition \ref{prop:condProb}}
\label{sec:condProbProof}

We start by defining several objects that are used in the proof. First note that $Q$ has unnormalised density $f$ with respect to the measure
\begin{equation} \label{eq:Qdom}
    \sum_{b \in \{0,1\}^{2k}} \lambda_b,
\end{equation}
where $\lambda_b$ is $\|b\|_0$-dimensional Lebesgue measure on $V_b := \{x \in \Omega_{2k} : x_i^0 = b_i \}$ and where as usual $0^0=0$. We parameterise $V_b$ with $\gamma_b: (0,\infty)^{d} \to V_b$, where $d := \|b\|_0$. Letting $(\sigma_1,\ldots,\sigma_d)$ be the ordered vector formed from $\{i : b_i=1\}$, we let
\begin{equation}
    \gamma_b(c) := c_1 e_{\sigma_1} + c_2 e_{\sigma_2} + \ldots + c_d e_{\sigma_d},
\end{equation}
where $e_i$ is the $i$\textsuperscript{th} standard basis vector for $\mathbb{R}^{2k}$. Finally, we we will also use the inclusion map $\iota_b: V_b \xhookrightarrow{} \Omega_{2k}$.

Our strategy is to partition $\Omega_{2k}$ into sets over which \eqref{eq:disintegration} can be verified. To do this, first let $O$ and $E$ be the set of all binary vectors of length $2k$ with at least one zero on the odd/even index, and no zeros on the even/odd index respectively. For example, $a \in O$ if and only if $a_{2i} = 1$ for all $i$, and $a_{2i+1} = 0$ for some $i$. Fix any $a \in O$ and $a' \in E$ and define 
\begin{equation} \label{eq:partition}
\Omega_{a,a'} := \{x \in \Omega_{2k}: x_i \leq \min_{j} \{x_{2j+1}\} \text{ iff } a_i = 0 \text{ and } x_i \leq \min_{j} \{x_{2j}\} \text{ iff } a'_i = 0\}.
\end{equation}
This set consists of vectors whose smallest odd elements coincide with the indices at which $a$ is zero, and whose smallest even elements coincide with the indices at which $a'$ is zero. A little thought shows that the sets \eqref{eq:partition} form a partition of $\Omega_{2k}$. We also define the pushforward of these sets by $\mathcal{T}_{a,a'} := T(\Omega_{a,a'})$.

Lemma \ref{lemma:resultonsubsets} shows that \eqref{eq:disintegration} holds when applied to these subsets. Most of our work will be in proving this lemma. 
\begin{lemma} \label{lemma:resultonsubsets}
Fix any set $\Omega_{a,a'}$ as in \eqref{eq:partition}. Then if $g: \Omega_{2k} \to \mathbb{R}$ is non-negative and measurable then for each $t \in \mathcal{T}_{a,a'}$, we have that $t \mapsto \int g(w) Q_t(\diff w)$ is measurable and
\begin{equation*}
\int_{\Omega_{a,a'}} g(x) Q(\diff x) = \int_{\mathcal{T}_{a,a'}}\int g(x) Q_t(\diff x)TQ(\diff t).
\end{equation*}
\end{lemma}
Lemma \ref{lemma:resultonsubsets} makes the proof of Proposition \ref{prop:condProb} straightforward. We first prove the proposition assuming the lemma, and then proceed to prove the lemma.
\newline

\begin{proof} [Proof of Proposition \ref{prop:condProb}]
First recall that the sets \eqref{eq:partition} form a measurable partition of $\Omega_{2k}$. Moreover the sets $\mathcal{T}_{a,a'}$ are disjoint. To see this, first fix some $t \in \mathcal{T}_{a,a'}$. Then there exists $x \in \Omega_{a,a'}$ such that $T(x) = t$. Any other point in $\{T=t\}$ takes the form $L_t(\Delta)$ for some $\Delta$ as defined in  \eqref{eq:lineSegmentParam}. Such points must also lie in $\Omega_{a,a'}$. The overall result then simply follows from additivity of measure over disjoint measurable sets, and applying Lemma \ref{lemma:resultonsubsets}.
\begin{align*} 
\int g(x) Q(\diff x) &= \sum_{a\in O} \sum_{a' \in E} \int_{\Omega_{a,a'}} g(x) Q(\diff x) \\
& = \sum_{a\in O} \sum_{a' \in E} \int_{\mathcal{T}_{a,a'}}\int g(x) Q_t(\diff x)TQ(\diff t)\\
& = \int \int g(x) Q_t(\diff x)TQ(\diff t).
\end{align*}
\end{proof}

\begin{proof}[Proof of Lemma \ref{lemma:resultonsubsets}]

Recall that $\Omega_{a,a'}$ consists of all $x \in \Omega_{2k}$ whose smallest odd elements coincide with the indices at which a is zero, and whose smallest even elements coincide with the indices at which $a'$ is zero. If both the smallest odd and even element of $x$ are zero, then $x \in V_{a''}$, where $a'' := a \odot a'$ and where the operator $\odot$ denotes component-wise multiplication. If the smallest odd (even) element is zero, but smallest even (odd) is positive then $x \in V_a$ ($x \in V_{a'}$). If all elements are positive then $x \in V_{1}$, where $1$ is the unit vector of length $2k$. This shows that $\Omega_{a,a'}$ is partitioned by its intersection with $V_a$, $V_{a'}$, $V_{a''}$ and $V_{1}$. For notational convenience, we let $\Omega := \Omega_{a,a'}$ and $\mathcal{T} := \mathcal{T}_{a,a'}$ for what follows. 

In particular, it is easy to see that $V_{a''} \subset \Omega_{a,a'}$. Letting $\mathcal{T}_1 = T(V_{a''})$, we first demonstrate that 
 \begin{equation} \label{eq:firstcond}
\int_{V_{a''}} g(x) Q(\diff x) = \int_{\mathcal{T}_1}\int g(x) Q_t(\diff x)TQ(\diff t),
\end{equation}
and then
\begin{equation} \label{eq:secondcond}
\int_{\Omega \setminus V_{a''}} g(x) Q(\diff x) = \int_{\mathcal{T}\setminus \mathcal{T}_1}\int g(x) Q_t(\diff x)TQ(\diff t).
\end{equation}
The lemma then follows trivially follows from summing both sides of \eqref{eq:firstcond} and $\eqref{eq:secondcond}$.

To verify \eqref{eq:firstcond}, observe that $a''$ must have at least one zero on both its odd and even side. Fix any $x \in V_{a''}$. The parameterization \eqref{eq:lineSegmentParam} shows that the level set $\{T = T(x) \}$ consists only of $x$ itself, implying that $x = u_{T(x)} = v_{T(x)}$ for any $x \in V_{a''}$. Therefore
\begin{align*}
\int_{V_{a''}} g(x) Q(\diff x) &= \int_{V_{a''}} g(u_{T(x)}) Q(\diff x)\\
&= \int_{\mathcal{T}_1} g(x_l) TQ(\diff t) \\
&= \int_{\mathcal{T}_1}\int g(x) Q_t(\diff x)TQ(\diff t),
\end{align*}
where the second step uses a change of variables $t = T(x)$ and the third uses the definition of $Q_t$ in \eqref{eq:conddistless}.

Proving \eqref{eq:secondcond} is slightly more involved. First recall that points in $\Omega \setminus V_{a''}$ must lie in exactly one of $V_a$, $V_a'$ or $V_1$. Also recall that $Q$ has density $f$ with respect to \eqref{eq:Qdom}. Therefore
\begin{align}
\begin{split}
    \int_{\Omega \setminus V_{a''}}g(x)Q(\diff x) &= \int_{\Omega \cap V_{a}}g(x)f(x)\lambda_a(\diff x) + \int_{\Omega \cap V_{a'}}g(x)f(x)\lambda_{a'}(\diff x) \\ &+ \int_{\Omega \cap V_{1}}g(x)f(x)\lambda_1(\diff x),
\end{split}
\end{align}
where we have removed null integrals that result from distributing \eqref{eq:Qdom}. Define $f_b := f \circ \gamma_b$ and $g_b := g \circ \gamma_b$ for an arbitrary vector $b \in \{0,1\}^{2k}$. This allows us to write
\begin{equation} \label{eq:remove_measures}
    \int_{\Omega \setminus V_{a''}} g(x) Q(\diff x) = \int_{U_a} g_a(p) f_a(p) \diff p  + \int_{U_{a'}} g_{a'}(q) f_{a'}(q) \diff q + \int_{\Omega \cap V_1} g(x) f(x)\diff x,
\end{equation}
where $U_a := \{p: \gamma_a(p) \in \Omega\}$ and $U_{a'} := \{q: \gamma_{a'}(q) \in \Omega \}$. We must be careful in interpreting each integral on the right hand side of \eqref{eq:remove_measures}. The vector $p$ for example is of length $\|a\|_{0}$, while $q$ is of length $\|a'\|_{0}$.

We split the proof into three cases. 

\noindent \textbf{Case 1, }$\mathbf{\|a\|_{0}=\|a'\|_{0}= 2k-1:}$  This implies that $a_i = 0$ for exactly one odd $i$. Consider the function $T_a:\mathbb{R}^{2k-1} \to \mathbb{R}^{2k-1}$ given by
\begin{equation*}
    T_a(p) := (p_1 + p_2, \ldots,p_{i-2} + p_{i-1}, p_{i-1}, p_{i}, p_{i}+p_{i+1}, \ldots, p_{2k-2} + p_{2k-1})^t.
\end{equation*}
This function satisfies $T_a(p) = T \circ \iota_a \circ \gamma_a(p)$ for all $p \in (0,\infty)^{2k-1}$. It is linear, non-singular and its Jacobian $J$ has determinant one. To see this, observe that $J$ is block diagonal with matrices $J_1$ and $J_2$
where $J_1$ has dimension $i-1$. $J_1$ is upper triangular and $J_2$ is lower triangular, and both matrices have ones along the diagonal. Therefore
\begin{equation*}
\det(J) = \det(J_1) \times \det(J_2) = 1 \times 1 = 1.
\end{equation*}
Applying the change of variables formula to the integral gives
\begin{align*}
\int_{U_a} g_a(p)f_a(p)\diff p &= \int_{T(\Omega \cap V_a)}g(\gamma_a(T_a^{-1}(t)))f(\gamma_a(T_a^{-1}(t))) \diff t\\
&= \int_{{T(\Omega \cap V_a)}}g(x_l)f(x_l) \diff t.
\end{align*}
Precisely the same argument gives an analogous form for the second integral on the right of \eqref{eq:remove_measures}, but with $x_u$ in the integrand rather than $x_l$, and with $T(\Omega \cap V_{a'})$ as the integration range.

We deal with the final integral in \eqref{eq:remove_measures} using a transformation $T_1:\mathbb{R}^{2k} \to \mathbb{R}^{2k}$ defined by
\begin{equation*}
    T_1(x) := (x_1, x_1+x_2, x_2+x_3, \ldots, x_{2k-1}+x_{2k})^t.
\end{equation*}
This is again linear and non-singular. The Jacobian is lower triangular with ones along the diagonal, and so the determinant is one. Using change of variables
\begin{align*}
\int_{\Omega \cap V_{1}} g(x)f(x) \diff x &= \int_{T(\Omega \cap V_{1})}\left(\int g(T_{1}^{-1}(x_1,t)) f(T_{1}^{-1}(x_1,t)) \diff x_1\right) \diff t\\
&= \int_{T(\Omega \cap V_{1})} \left( \frac{1}{\sqrt{2k}} \int_{L_t} g(x) f(x) \diff s \right) \diff t\\
&= \int_{T(\Omega \cap V_{1})} \left( \frac{1}{\sqrt{2k}} \int_{L_t} f(x) \diff s \frac{\int_{L_t} g(x) f(x) \diff s}{\int_{L_t} f(x) \diff s} \right) \diff t\\
&= \int_{T(\Omega \cap V_{1})} \left(\alpha_t \int g(x) \mu_t(\diff x) \right) \diff t.
\end{align*}
In the second step the inner integral is rewritten as a line integral over $L_t$. Then, it is written in a form that allows application of the definition of $\mu_t$ in \eqref{eq:line_measure}. 

Putting this all together and grouping the integrals gives
\begin{align}
\int_{\Omega_{a,a'} \setminus V_{a''}}& g(x) P(\diff x)\\
&=\int_{\mathcal{T}^1_{a,a'}} \biggl( g(x_l)f(x_l) + g(x_u)f(x_u) + \alpha_t \int g(x) \mu_t(\diff x)\biggr) \diff t \label{eq:grouped}\\
&=  \int_{\mathcal{T}^1_{a,a'}} \int g(x) P_t(\diff x)\left( f(x_l) + f(x_u) + \alpha_t \right) \diff t \label{eq:subpt}\\
&=  \int_{\mathcal{T}^1_{a,a'}} \int g(x) P_t(\diff x) \biggl( f(x_l)+ f(x_u) + \int f(x_1,t)\diff x_1 \biggr) \diff t \\
&=  \int_{\mathcal{T}^1_{a,a'}} \int g(x) P_t(\diff x) \left( TP_a + TP_{a'} +TP_1\right)(\diff t) \\
&=  \int_{\mathcal{T}^1_{a,a'}} \int g(x) P_t(\diff x)TP(\diff t).
\end{align}
\eqref{eq:subpt} is obtained from \eqref{eq:grouped} by evaluating the integral of $g(x)$ with respect to $P_t$, where $P_t$ is defined in Proposition \ref{prop:condProb}.\\

\noindent \textbf{Case 2, } $\mathbf{\|a\|_{0}=\|a'\|_{0} < 2k-1:}$ Fix some $x \in \Omega \setminus V_{a''}$. Either the smallest even element, the smallest odd element, or both, are positive. Since the smallest element appears at more than one index, there must be ties between positive elements. Therefore, $\Omega \setminus V_{a''}$ must be a null set under $Q$, and so $Q_t$ may be defined arbitrarily on $\mathcal{T} \setminus \mathcal{T}_{1}$.\\

\noindent \textbf{Case 3, } $\mathbf{\|a\|_{0} \neq \|a'\|_{0}:}$ Now suppose that $\|a\|_{0} < \|a'\|_{0}$ and assume, for the moment, that the second and third integrals on the right hand side of \eqref{eq:remove_measures} are null over $\Omega$. Now 
\begin{align*}
\int_{\Omega_{a,a'}} g(x) P(\diff x) &= \int_{\Omega_{a,a'}} g(u_{T(x)}) P(\diff x) && \text{(using that $x = u_{T(x)}$ on $V_a$)}\\
&= \int_{\mathcal{T}_{a,a'}} g(x_l) TP(\diff t) && \text{(change of variables)}\\
&= \int_{\mathcal{T}_{a,a'}}\int g(x) P_t(\diff x)TP(\diff t) && \text{(definition of $P_t$)},
\end{align*}
as required. The analogous result is established for $\|a'\|_{l_0} < \|a\|_{l_0}$ using identical reasoning. \\
\end{proof}

\section{Appendix B: Proof of Proposition \ref{prop:inadmissible_null}}

\begin{proof}[\textbf{Proof of Proposition \ref{prop:inadmissible_null}}]
Begin by verifying the first statement; that the set of graphs producing inadmissible data is $\Lambda$-negligible. Fix some data $(d,s)$ and substitute the definition of $s^+_u$ and $s^+_v$ in \eqref{eq:in_out_match} to get 
\begin{align}
    \sum_{u \in U_i} \sum_{v \in N} w_{uv} &= \sum_{u \in N} \sum_{v \in V_i} w_{uv} \\
    \sum_{u \in U_i} \sum_{v \in N \setminus V_i} w_{uv} &= \sum_{u \in N \setminus U_i} \sum_{v \in V_i} w_{uv} \label{eq:posequal},
\end{align}
for each $i \in I$ and any $G \in \mathcal{G}_m(d,s)$. In the second step, we have simply removed summands common to both sides. If \eqref{eq:posequal} is positive on either side then it implies a disjoint sum of edge weights exactly equate. This event is $\Lambda$-negligible, and so it suffices to show that any graph producing inadmissible data must have \eqref{eq:posequal} positive for some $i \in I$.

First suppose that condition 1 of Definition \ref{def:admissible} is violated. Then $(U_i)_{i \in I}$ and $(V_i)_{i \in I}$ can be labeled so that $U_1 \cap U_2 \neq \emptyset$ and/or $V_1 \cap V_2 \neq \emptyset$. We show by contradiction that \eqref{eq:posequal} must be positive for some $i \in \{1,2\}$. Suppose first that \eqref{eq:posequal} is zero for $i \in \{1,2\}$. Then \eqref{eq:posequal} holds for $U' = U_1 \setminus U_2$ and $V' = V_1 \setminus V_2$. We show this by expanding the left hand side of \eqref{eq:posequal}
\begin{align}
    \sum_{u \in U_1} \sum_{v \in N \setminus V_1} w_{uv} &= \sum_{u \in U_1 
    \setminus U_2} \sum_{v \in N \setminus V_1} w_{uv} \\
    &= \sum_{u \in U'} \left( \sum_{v \in N \setminus V'} w_{uv} - \sum_{v \in V_2 \cap V_1} w_{uv} \right) \label{eq:expandsetminus}\\
    &= \sum_{u \in U'}\sum_{v \in N \setminus V'} w_{uv} - \sum_{u \in U'}\sum_{v \in V_2 \cap V_1} w_{uv} \label{eq:removefinal}\\
    &= \sum_{u \in U'}\sum_{v \in N \setminus V'} w_{uv} = 0.
\end{align}
In \eqref{eq:expandsetminus} we have used that $N \setminus V' = (V_1 \cap V_2) \cup (N \setminus V_1)$. To see how we remove the final summation in \eqref{eq:removefinal}, observe that if $u \in U'$ then $u \in N \setminus U_2$. Also if $v \in V_1 \cap V_2$ then $v \in V_2$. Therefore because \eqref{eq:posequal} holds for $i=2$ and by assumption is equal to zero, this summation must also be zero. The same reasoning as above can be applied to the right hand side of \eqref{eq:posequal} to show that 
\begin{equation*}
    \sum_{u \in U'} \sum_{v \in N \setminus V'} w_{uv} = \sum_{u \in N \setminus U'} \sum_{v \in V'} w_{uv} = 0,
\end{equation*} 
which implies that $U'$ and $V'$ satisfy \eqref{eq:in_out_match}. Since $U' \times V' \subset U_1 \times V_1$ this contradicts the definition of $U_1$ and $V_1$, and establishes that graphs for such data have positive ties, and thus lie in a $\Lambda$-negligible set.

Now suppose that the second condition of Definition \ref{def:admissible} is violated, i.e. $\tilde{\mathcal{G}}_m(d,s)$ is empty. This could be because the reference set \eqref{eq:targetspace} is empty, in which case no graph in $\mathcal{G}$ aligns with the data $(d,s)$. Suppose instead that \eqref{eq:targetspace} is non-empty. Any graph in \eqref{eq:targetspace} has some edge $uv$ such that
\begin{equation*}
uv \notin \bigcup_{i\in I} \left(U_i \times V_i\right).
\end{equation*}
This implies that \eqref{eq:posequal} is positive for some $i^* \in I$ and that the graph must have positive ties in its weight matrix. Therefore the set of such graphs is $P$-null. Moreover, this argument verifies the last statement in the proposition, which is that the set of graphs not in $\tilde{\mathcal{G}}_m(d,s)$ for some data $(d,s)$ is $P$-null.
\end{proof}

\section{Appendix C: Proof of Proposition \ref{prop:irreducibilty}}
\label{sec:irreducibilityproof}

Our strategy is to first establish open set irreducibility with respect to some topology on $\mathcal{G}_m(d,s)$. This result is stated in Lemma \ref{lemma:neighbourhood}. This is linked to $\psi$-irreducibility, which then establishes Proposition \ref{prop:irreducibilty}.

First we define objects used in the proofs. Fix some admissible data $(d,s)$, as in the proposition. Associate each graph $G \in \mathcal{G}_m(d,s)$ with a weighted, bipartite and undirected graph $B(G) := (R,C,W)$, with vertex sets $R := \{r_u : u \in N\}$ and $C := \{c_u : u \in N\}$, and with weight matrix defined through $w_{r_uc_v}(B) := w_{uv}(G)$. Recall the vertex sets $\{U_i\}_{i\in I}$ and $\{V_i\}_{i \in I}$ introduced in Section \ref{sec:stochastic_stability}, which are associated with the data $(d,s)$. We define sets
\begin{equation} \label{eq:ei}
    E_i := \{p_u : u \in U_i\} \cup \{q_v : v \in V_i\},
\end{equation}
for each $i \in I$. By the definition of admissibility (Definition \ref{def:admissible}) it is clear that these sets form a partition of $P \cup Q$.

We begin by stating and proving two lemmas which will help establish Lemma \ref{lemma:neighbourhood}.

\begin{lemma}
Fix some $i \in I$. The vertex set $E_i$ is connected in $B(G)$ for all $G \in \mathcal{G}_m(d,s)$.
\label{lemma:connectedH}
\end{lemma}

\begin{proof}
Fix any $G \in \mathcal{G}_m(d,s)$ and let $C$ be a connected component of $B(G)$. It is easy to see that 
\begin{equation*}
    \sum_{u \in \{u : p_u \in C \}} s^{-}_u =  \sum_{v \in \{v : q_v \in C \}} s^{+}_v,
\end{equation*}
and so, by the definition of admissibility, $C$ must be the union of sets of the form \eqref{eq:ei}. This implies that $E_i$ must wholly lie within a connected component for all graphs in the reference set.
\end{proof}

The next lemma establishes that the Markov chain can move between different graph topologies. In particular, it shows that an edge can be added without altering the rest of the graph's topology. It will be used in the proof of Lemma \ref{lemma:neighbourhood}.

\begin{lemma} \label{lemma:posedge}
Let $uv \in U_i \times V_i$ for some $i \in I$, and $uv \notin \mathcal{F}$. Suppose the current state of the chain is $G$ and $uv \notin E(G)$. Fixing $\epsilon > 0$, these is positive probability of reaching some $G^*$ satisfying $W_{uv}(G^*) \in (0,\epsilon]$ and $E(G^*) = E(G) \cup \{uv\}$ in one iteration.
\end{lemma}

\begin{proof}[Proof of Lemma \ref{lemma:posedge}]
Since $uv \in U_i \times V_i$, $p_u$ and $q_v$ must belong to $E_i$. By Lemma \ref{lemma:connectedH}, $p_u$ and $q_v$ must be connected by a simple path in $B(G)$. Because $B(G)$ is bipartite, the path must have odd length, and when including $p_uq_v$, defines a $k$-cycle with one zero entry and no forced edges. The $k$-cycle selection strategy defined in Algorithm \ref{alg:selectkcycle} gives positive probability to all such $k$-cycles. Sampling this $k$-cycle and $\Delta \in (0,\epsilon]$ yields a graph with the required properties. Sampling $\Delta$ in this range has positive probability because the density $f$ is assumed to be positive everywhere. 
\end{proof}

Lemma \ref{lemma:neighbourhood} shows that the Markov chain is open set irreducible, where the open sets are those induced by the metric

\begin{equation} \label{eq:metric}
    d(G_1,G_2) := \max_{uv \in N^2}\mid w_{uv}(G_1) - w_{uv}(G_2) \mid + \rho(G_1,G_2),
\end{equation}
where $\rho(G_1,G_2) = 1$ if $G_1$ and $G_2$ hae differing topologies, and is otherwise zero. It is easy to check that this really is a metric. Before stating the lemma, we recall the definition of $\tilde{\mathcal{G}}_m(d,s)$ from Section \ref{sec:stochastic_stability}.

\begin{lemma} \label{lemma:neighbourhood}
Let the current state of the chain be $G \in \mathcal{G}_m(d,s)$, and fix any $G' \in \tilde{\mathcal{G}}_m(d,s)$, and $\epsilon > 0$. There exists some integer $n$ for which there is positive probability of reaching an $\epsilon$-neighbourhood (under \eqref{eq:metric}) of $G'$ within $n$ steps.
\end{lemma}

\begin{proof}[Proof of Lemma \ref{lemma:neighbourhood}]
Form a signed graph $H := (N,D)$, where $D = (d_{uv})$ is a matrix of possibly negative weights that satisfy $d_{uv} := w_{uv}(G) - w_{uv}(G')$. Label $uv$ \textit{red} if $d_{uv} > 0$ and \textit{blue} if $d_{uv} < 0$. $G$ and $G'$ are equal if and only if $E(H)$ is empty. Red edges must be in $E(G)$, however blue edges may not be in $E(G)$. Therefore, begin by using Lemma \ref{lemma:posedge} repeatedly to adjust $G$ so that all blue edges in $H$ are in $E(G)$. This Lemma can be applied because we have assumed that $G' \in \tilde{\mathcal{G}}_m(d,s)$. This implies that if $uv$ is blue then it must belong to $U_i \times V_i$ for some $i \in I$, because otherwise $w_{uv}(G') = 0$, which contradicts the requirement that $w_{uv}(G) < w_{uv}(G')$ for blue edges.

Call a $k$-cycle \textit{alternating} if its vertex pairs \eqref{eq:coords} alternate between red and blue edges , when interpreted as part of $H$. As long as $E(H)$ is non-empty, one can always form an alternating $k$-cycle. To see this, note that the in- and out-strengths of vertices in $H$ are uniformly zero. Therefore, if $uv$ is red ($d_uv > 0)$, then there must exist blue $wv$ for which $d_{wv} < 0$. Hence walk along alternating edges until returning to a vertex for the first time, forming an alternating $k$-cycle. 

Fix one such $k$-cycle $z_1$ ordered so that the first edge is \textit{red}. All edges in the cycle are positive and do not belong to $\mathcal{F}$. Let
\begin{equation*}
    \Delta'_1 := - \min_{uv \in z_1}\{\|D_{uv}\|\}
\end{equation*}
along the cycle. If we sampled $\Delta = \Delta_1$ exactly along $z$, this would remove an edge from $E(H)$. Repeating the process at most $d$ times, where $d$ is the size of $E(h)$, yields $k$-cycles $z_1,\ldots,z_d$ and $\Delta'_1,\ldots, \Delta'_d$, after which we reach $G'$. Therefore there must exist some $\epsilon_0 > 0$ such that sampling $\Delta_i \in \Delta'_i + [-\epsilon_0,\epsilon_0]$ sequentially along these cycles gives a graph in an $\epsilon$-neighbourhood of $G'$.
\end{proof}

To link open set irreducibility (Lemma \ref{lemma:neighbourhood}) with $\psi$-irreducibility, it is helpful to view the reference set $\mathcal{G}$ as a subset of a vector space. This will provide a geometric interpretation to the $k$-cycles that form the basis of our Markov chain. This will motivate a measure $\varphi$ on $\mathcal{G}$ for which it is easy to demonstrate $\varphi$-irreducibility of the chain.

Let $\mathbb{R}^{N \times N}$ be the vector space of real-valued $N \times N$ matrices. Equip this with the metric defined by \eqref{eq:metric}. We will assume that if $uv \notin U_i \times V_i$ for some $i \in \mathcal{I}$ then $uv \in \mathcal{F}$. Let $V$ be the affine subspace of $\mathbb{R}^{N\times N}$ with row and column margins $s^-$ and $s^+$ respectively and additionally respecting the forced zeros implied by $\mathcal{F}$. This has some dimension 
\begin{equation*}
d \geq N^2 + 1 - 2N - \mid \mathcal{F} \mid.
\end{equation*}
Then the reference set satisfies
\begin{equation*}
    \mathcal{G} = V \cap \Omega_{N \times N},
\end{equation*}
where $\Omega_{N \times N}$ is the $N \times N$-dimensional non-negative orthant.

Fixing $W \in \mathcal{G}$, we see that a $k$-cycle is equivalent to sampling from a line in $V$. The `direction' of this line is given by an $N \times N$ matrix $M$, where $M_{uv} = 1$ if $uv$ is on the odd side of the cycle, $M_{uv}=-1$ if on the even side, and with all other entries being zero. Fix any $W^*$ in $V$. The arguments used in the proof of \ref{lemma:neighbourhood} can easily be extended to show that there exists $(M^*_1,\ldots, M^*_{L})$ and a real-valued vector $(\Delta_1^*,\ldots,\Delta_L^*)$ for which
\begin{equation*}
    W^* = W + \sum_{l=1}^L M^*_l \Delta^*_l,
\end{equation*}
and each $M^*_l$ corresponds to a $k$-cycle. This in turn implies that there exists a set of such matrices labelled $(M_1,\ldots, M_d)$ which form an affine basis for $V$. Therefore 
\begin{equation}
    f_W(\Delta) := W + \sum_{i=1}^d M_i \Delta_i,
\label{eq:param}
\end{equation}
where $\Delta := (\Delta_1,\ldots, \Delta_d)$ parameterizes $V$. Equation \eqref{eq:param} is a homeomorphism between $\mathbb{R}^d$ and $V$. We are now ready to prove Proposition \ref{prop:irreducibilty}.
\newline

\begin{proof}[Proof of Proposition \ref{prop:irreducibilty}]
Fix any $W \in V$ for which $W_{uv} > 0$ if $uv \notin \mathcal{F}$. The existence of such a $W$ is implied by Lemma \ref{lemma:posedge}. Consider the parameterization of $V$ defined by \eqref{eq:param}. We use this to define a measure $\varphi$ on $\mathcal{G}$, and show that the Markov chain is irreducible with respect to $\varphi$.

Let $\epsilon_0 := \min_{uv\notin\mathcal{F}} \mid w_{uv} \mid / d$, and define $N_{\epsilon_0} := (-\epsilon_0,\epsilon_0)^d \subset \mathbb{R}^d$. Let $\mu$ be Lebesgue measure restricted to $N_{\epsilon_0}$. Define $\varphi$ on $(\mathcal{G},\mathcal{B})$ as the pushforward of $\mu$ under $f_w$, so that
\begin{equation*}
    \varphi(A) = \lambda^d\left(f^{-1}(A) \cap N_{\epsilon_0}\right),
\end{equation*}
for measurable $A$.

We now show $\varphi$-irreducibility. Fix any measurable $A$ for which $\varphi(A) > 0$. Letting $E := f_w^{-1}(A) \cap N_{\epsilon_0}$, it is clear that $E$ must be Lebesgue positive in $\mathbb{R}^d$. Consider a Markov chain $X_n$ with kernel $\Lambda$. Define $\Phi_n := f_w^{-1}(X_n)$ and suppose $\Phi_n \in N_{\epsilon_0}$. Each of the basis matrices $m_1, \ldots, m_d$ corresponds to a $k$-cycle. There is positive probability that the chain selects the cycle corresponding to $m_k$ at the $n+k$\textsuperscript{th} step. Conditional on this, $\varphi_{n+d}$ has positive density everywhere on $N_{\epsilon_0}$. This implies that if $X_n \in f_w(N_{\epsilon_0})$ then $Q^d(X_n, A) > 0$.

It remains to show that for each $x$, $Q^n(x, f_w(N_{\epsilon_0})) > 0$ for some $n$. Since $f_w$ maps open sets, $f_w(N_{\epsilon_0})$ is open in $\mathcal{G}$. Therefore this result follows from Lemma \ref{lemma:neighbourhood}.
\end{proof}

\end{document}

%% file: figures/matrix_drawing.tex
\begin{subfigure}{0.3\linewidth}
\centering
\scalebox{0.6}{
\begin{tikzpicture}
\draw[color = darkgray] (-3,-3) grid (3,3);
\filldraw[fill=gray!20] (-3,2) rectangle (-2,3);
\filldraw[fill=gray!20] (-2,1) rectangle (-1,2);
\filldraw[fill=gray!20] (-1,0) rectangle (0,1);
\filldraw[fill=gray!20] (0,-1) rectangle (1,0);
\filldraw[fill=gray!20] (1,-2) rectangle (2,-1);
\filldraw[fill=gray!20] (2,-3) rectangle (3,-2);
\filldraw[fill=red!20, draw=black] (-2,-3) rectangle (-1,-2);
\node at (-1.5,-2.5) {\scalebox{1.5}{$-\Delta$}};
\filldraw[fill=red!20, draw=black] (-2,1) rectangle (-1,2);
\node at (-1.5,1.5) {\scalebox{1.5}{$+\Delta$}};
\filldraw[fill=red!20, draw=black] (2,-3) rectangle (3,-2);
\node at (2.5,-2.5) {\scalebox{1.5}{$+\Delta$}};
\filldraw[fill=red!20, draw=black] (2,1) rectangle (3,2);
\node at (2.5,1.5) {\scalebox{1.5}{$-\Delta$}};
\end{tikzpicture}
}
\captionsetup{justification=centering}
\caption{\footnotesize Along $\{22,26,66,62\}$.}
\label{fig:2cycle}
\end{subfigure}
\begin{subfigure}{0.3\linewidth}
\centering
\scalebox{0.6}{
\begin{tikzpicture}
\draw[color = darkgray] (-3,-3) grid (3,3);
\filldraw[fill=gray!20] (-3,2) rectangle (-2,3);
\filldraw[fill=gray!20] (-2,1) rectangle (-1,2);
\filldraw[fill=gray!20] (-1,0) rectangle (0,1);
\filldraw[fill=gray!20] (0,-1) rectangle (1,0);
\filldraw[fill=gray!20] (1,-2) rectangle (2,-1);
\filldraw[fill=gray!20] (2,-3) rectangle (3,-2);
\filldraw[fill=red!20, draw = black] (-2,-1) rectangle (-1,0);
\node at (-1.5,-0.5) {\scalebox{1.5}{$-\Delta$}};
\filldraw[fill=red!20, draw = black] (-2,0) rectangle (-1,1);
\node at (-1.5,0.5) {\scalebox{1.5}{$+\Delta$}};
\filldraw[fill=red!20, draw = black] (1,-1) rectangle (2,0);
\node at (1.5,-0.5) {\scalebox{1.5}{$+\Delta$}};
\filldraw[fill=red!20, draw = black] (1,2) rectangle (2,3);
\node at (1.5,2.5) {\scalebox{1.5}{$-\Delta$}};
\filldraw[fill=red!20, draw = black] (2,0) rectangle (3,1);
\node at (2.5,0.5) {\scalebox{1.5}{$-\Delta$}};
\filldraw[fill=red!20, draw = black] (2,2) rectangle (3,3);
\node at (2.5,2.5) {\scalebox{1.5}{$+\Delta$}};
\end{tikzpicture}
}
\captionsetup{justification=centering}
\caption{\footnotesize Along $\{45,42,32,36,16,15\}$}\label{fig:3cycle}
\end{subfigure}
\begin{subfigure}{0.3\linewidth}
\centering
\scalebox{0.6}{
\begin{tikzpicture}
\draw[color = darkgray] (-3,-3) grid (3,3);
\filldraw[fill=gray!20] (-3,2) rectangle (-2,3);
\filldraw[fill=gray!20] (-2,1) rectangle (-1,2);
\filldraw[fill=gray!20] (-1,0) rectangle (0,1);
\filldraw[fill=gray!20] (0,-1) rectangle (1,0);
\filldraw[fill=gray!20] (1,-2) rectangle (2,-1);
\filldraw[fill=gray!20] (2,-3) rectangle (3,-2);
\filldraw[fill=red!20, draw = black] (-3,-3) rectangle (-2,-2);
\node at (-2.5,-2.5) {\scalebox{1.5}{$-\Delta$}};
\filldraw[fill=red!20, draw = black] (-3,1) rectangle (-2,2);
\node at (-2.5,1.5) {\scalebox{1.5}{$+\Delta$}};
\filldraw[fill=red!20, draw = black] (-2,-2) rectangle (-1,-1);
\node at (-1.5,-1.5) {\scalebox{1.5}{$+\Delta$}};
\filldraw[fill=red!20, draw = black] (-2,1) rectangle (-1,2);
\node at (-1.5,1.5) {\scalebox{1.5}{$-\Delta$}};
\filldraw[fill=red!20, draw = black] (1,-3) rectangle (2,-2);
\node at (1.5,-2.5) {\scalebox{1.5}{$+\Delta$}};
\filldraw[fill=red!20, draw = black] (1,0) rectangle (2,1);
\node at (1.5,0.5) {\scalebox{1.5}{$-\Delta$}};
\filldraw[fill=red!20, draw = black] (2,-2) rectangle (3,-1);
\node at (2.5,-1.5) {\scalebox{1.5}{$-\Delta$}};
\filldraw[fill=red!20, draw = black] (2,0) rectangle (3,1);
\node at (2.5,0.5) {\scalebox{1.5}{$+\Delta$}};
\end{tikzpicture}
}
\captionsetup{justification=centering}
\caption{\footnotesize Along $\{36,35,65,61,21,22,52,56\}$.}\label{fig:4cycle}
\end{subfigure}